\documentclass[12pt, a4paper]{article}
\usepackage[utf8]{inputenc}
\usepackage{ragged2e}
\usepackage{amsmath,amsfonts,amsthm,amssymb}
\usepackage{setspace}
\usepackage{fancyhdr}
\usepackage{lastpage}
\usepackage{extramarks}
\usepackage{chngpage}
\usepackage{soul,color}
\usepackage{graphicx,float,wrapfig}
\usepackage{listings}
\usepackage{tikz}
\usepackage{harpoon}
\usepackage{algorithm}
\usepackage{algpseudocode}

\usepackage{hyperref}
\usepackage{bm}
\usepackage{indentfirst}
\usepackage{natbib}
\usepackage{verbatim}
\usepackage{amsmath}
\usepackage{amssymb}
\usepackage{tikz}

\hypersetup{
    colorlinks=true,
    citecolor=blue,
    linkcolor=blue,
    urlcolor=blue,
    bookmarksopen=true,
    pdfstartview={XYZ null null 1.00},
    pdfpagelayout={SinglePage}
}

\newtheorem{theorem}{Theorem}
\newtheorem{lemma}{Lemma}
\newtheorem{proposition}{Proposition}
\newtheorem{corollary}{Corollary}
\newtheorem{definition}{Definition}

\newtheorem{example}{Example}
\newtheorem{assumption}{Assumption}

\usetikzlibrary{arrows,graphs}

\newcommand{\enabstractname}{Abstract}
\newenvironment{enabstract}{
  
  \noindent\mbox{}\hfill{\bfseries \enabstractname}\hfill\mbox{}\par
  \vskip 2.5ex}{\par\vskip 2.5ex}

\graphicspath{{figure/}}                                 
\usepackage[a4paper]{geometry}                           
\begin{document}
\justifying
\title{On the Equivalence of Synchronous Coordination Game and Asynchronous Coordination Design\thanks{This paper benefits from numerous suggestions and comments. I am grateful to (in random order) Yves Zenou, John Quah, Jinwoo Kim, Tilman B\"orgers,  Andrew Rhodes, Frank Yang, Giacomo Lanzani, Joshua Lanier, Xu Tan, Fan Wu, Rui Tang, Tat-how Teh, Yiqing Xing, Xingye Wu, Jie Zheng, and seminar participants at Tsinghua University, Renmin University, and 7th World Congress of Game Theory Society. For the newest version, please click \href{https://www.dropbox.com/scl/fi/gdu8m53806c0kqd71w0c9/Xinnian_Pan_On_The_Equivalence.pdf?rlkey=zjqtgdubjjfhre49yvv5ea09t&st=tu9r7aaq&dl=0}{here}.}}
\date{\today}
\author{Xinnian (Kazusa) Pan\thanks{\baselineskip=0.7\normalbaselineskip Tsinghua University. Email: panxn19@mails.tsinghua.edu.cn.} 
}
\maketitle

\begin{enabstract}
\linespread{1.3} \selectfont
This paper establishes the equivalence between synchronous and asynchronous coordination mechanisms in dynamic games with strategic complementarities and common interests. Synchronous coordination, characterized by simultaneous commitments, and asynchronous coordination, defined by sequential action timing, are both prevalent in economic contexts such as crowdfunding and fund management. We introduce Monotone Subgame Perfect Nash Equilibrium (MSPNE) to analyze least favorable equilibrium outcomes. We provide a recursive characterization for synchronous coordination and a graph-theoretic representation for asynchronous coordination, demonstrating their equivalence in terms of the greatest implementable outcome. Our results show that the structure of commitment—whether simultaneous or sequential—does not affect the achievable welfare outcome under certain conditions. Additionally, we discuss computational aspects, highlighting the general NP-Hardness of the problem but identifying a significant class of games that are computationally tractable. These findings offer valuable insights for the optimal design of coordination mechanisms.

\vspace{1ex}

\end{enabstract}

\newpage

\section{Introduction}
Coordination games play a crucial role in understanding various economic phenomena, such as financial market stability, technological adoption, and public goods provision. In these settings, individuals or firms must align their actions to achieve collectively beneficial outcomes. Examples include financial market participants coordinating their expectations to prevent systemic crises, firms deciding on technological standards to ensure compatibility, or countries negotiating environmental preservation legislation to achieve optimal levels of carbon emissions. However, the multiplicity of equilibria often complicates the coordination process, making it challenging for players to achieve socially optimal outcomes. This paper addresses this issue by analyzing both synchronous and asynchronous coordination settings within a dynamic game framework.

In the synchronous setting, players decide whether to commit to the socially efficient action simultaneously across multiple stages, while observing the history of prior actions. This setting is particularly relevant for understanding scenarios where coordination is structured and collective, such as simultaneous investment decisions by firms responding to technological changes or coordinated actions by decentralized depositors facing market shocks. The synchronous model highlights how commitment can foster coordination by reducing uncertainty about others' choices, thereby facilitating convergence to an efficient equilibrium. In contrast, the asynchronous setting allows players to act sequentially, with each player making a decision only once during the game and observing prior actions before committing. This setting captures economic environments where the timing of decisions varies among participants. In many cases, a principal may have some power to manipulate the structure of sequential moves, such as designing diversified redemption dates during an investment period. This sequential structure can help mitigate coordination failures by allowing players to condition their actions on observed behavior.

We focus on strategies that are monotone in the history of actions, and accordingly, the solution concept we adopt is Monotone Subgame Perfect Nash Equilibrium (MSPNE). MSPNE represents a subgame perfect Nash equilibrium that is monotone with respect to historical actions. The dynamic game we analyze is an extensive-form game with complete and almost perfect information, where SPNE is a standard solution concept. However, the requirement of monotonicity introduces subtle distinctions. By emphasizing these monotone strategies, we provide an elegant characterization of both synchronous and asynchronous games. In contrast, non-monotone strategies often exploit over-observed historical data to create credible but unreasonable threats, where the player making the threat cannot benefit from following through. We present numerous examples to illustrate the issues with such strategies (see \autoref{sec:counterexample}).

Our analysis offers new insights into the interplay between commitment timing and strategic complementarities. By explicitly modeling and solving both synchronous and asynchronous decision-making frameworks, we bridge the gap between static coordination problems and dynamic strategic interactions. This connection is crucial for understanding a wide range of economic problems, from coordinating microeconomic behavior efficiently to eliminating systemic risks in financial markets. The methods developed in this paper form a foundation for analyzing dynamic games with strategic complementarities. The recursive characterization of MSPNE and the use of tree-depth equivalence allow us to effectively navigate the complexities of both synchronous and asynchronous settings. By focusing on the worst-case, least favorable equilibrium outcome, our analysis employs a robust approach that is well-suited to addressing coordination challenges with significant welfare implications.

The primary goal of this paper is to establish an equivalence between synchronous and asynchronous coordination mechanisms. Despite their structural differences, we show that when the principal adopts a pessimistic stance and evaluates performance based on the worst equilibrium (MSPNE) outcome, the greatest implementable outcome achieved through an asynchronous partition design is equivalent to that obtained in the synchronous setting. This result implies that the maximal power of commitment, whether synchronous or asynchronous, has no significant difference on the achievable outcome under certain conditions. By providing a unified theoretical framework for both settings, this paper demonstrates that the coordination challenges inherent in dynamic games can be effectively addressed regardless of the timing structure, provided key conditions are met.

To establish this equivalence, we employ a novel approach using tree-depth characterization and graph-theoretical constructs to analyze the underlying structure of strategic interdependencies in these coordination games. We also highlight the bridging role of the weakest link game on a directed graph, which has received limited attention in the literature. By leveraging these tools, we comprehensively characterize the MSPNE of synchronous games and the least MSPNE of arbitrary asynchronous games, offering a deeper analysis of how the timing of actions—whether simultaneous or sequential—affects the ability to coordinate on efficient equilibria. We believe that these methodologies can be adapted to study a broader range of environments, even in cases involving incomplete information and imperfect observations.

Our findings have implications for economic policy and institutional design. For instance, in the context of financial regulation or technology adoption interventions, our results suggest that participants can achieve equally efficient coordination regardless of whether collective decisions are voluntary (decentralized) and synchronously made or instructed (centralized) and asynchronously made. This insight is particularly valuable in economies where coordination and strategic complementarities play a crucial role.

Furthermore, we explore the flexibility of the model and the robustness of our results from several perspectives. When the number of periods is fixed, the optimal asynchronous design is independent of the principal's preferences, provided the principal always weakly prefers higher equilibrium outcomes. As a result, the choice of the number of partitions in the asynchronous design problem can be simplified into a cost-benefit analysis framework, assuming the principal also strictly prefers fewer partitions. We further demonstrate that every optimal asynchronous design for some preferences aligns with a totally ordered weak centrality measure, which means the set of candidate designs is relatively small compared to the number of players. Although the game-theoretical analysis is NP-Hard for most stage games, we identify a significant class of games where the corresponding dynamic programming problem has an unambiguous policy function and can thus be solved efficiently: ordered games. These games illustrate the alignment of two orders, termed the cost order and the contribution order—players who have a greater influence on others' incentives also incur a higher initiation cost to take the high action. Lastly, we consider time-dependent payoffs in the synchronous game and show that all of our results continue to hold under standard regularity conditions, such as nearness and procrastination, which are satisfied by lag-complementarities flow payoffs, for instance.

Overall, this paper provides a comprehensive framework for understanding dynamic coordination under different timing structures. By focusing on the least MSPNE and establishing the equivalence between synchronous and asynchronous coordination designs, we contribute to the broader literature on dynamic games and strategic complementarities, offering new insights into how coordination can be achieved efficiently across diverse economic contexts. The implications of our findings extend beyond theoretical contributions, providing practical insights for policymakers, firms, and institutions that face the challenge of coordinating actions in complex, interdependent environments.

\subsection{Related Literature}
This paper contributes to the literature on dynamic coordination games, particularly those involving irreversibility and strategic complementarities. The foundational work of \citet{farrell1985standardization} introduced models of technology adoption with binary actions and irreversibility, laying the groundwork for understanding coordination in the context of technological standards. \citet{gale2001monotone} and \citet{matthews2013achievable} extended these ideas by formalizing dynamic coordination games with irreversibility, which form a basis for analyzing efficient coordination outcomes over time. Our approach builds on these earlier contributions by emphasizing the role of interdependent payoffs and expanding the scope to asynchronous coordination.

The closest work to ours is \citet{gale1995dynamic}, which shows that synchronous games can facilitate efficient coordination while exhibiting robust delay as a common feature. \citet{gale1995dynamic} focuses on the discount structure of these games, whereas our paper emphasizes the interdependent payoff structure and its implications for coordination efficiency. Additionally, the literature on dynamic games with both synchronous and asynchronous movements, such as \citet{admati1991joint}, \citet{lagunoff1997asynchronous}, \citet{marx2000dynamic}, \citet{lockwood2002gradualism}, \citet{duffy2007giving}, and \citet{dutta2012coordination}, provides a foundation for understanding how different timing mechanisms contribute to achieving socially efficient outcomes.

Experimental studies have also tested the efficiency of dynamic coordination. \citet{choi2008sequential}, \citet{choi2011network}, \citet{jin2023coordination}, and \citet{avoyan2023road} have explored how sequential decision-making and network effects impact coordination outcomes, providing empirical support for theoretical models of dynamic coordination games.

The literature on commitment games is also closely related to our work. Commitment games, as studied by \citet{bade2009bilateral}, \citet{renou2009commitment}, and \citet{dutta2016dynamic}, involve players deciding whether to pledge to specific actions during a commitment stage before making decisions in a subsequent stage. A key distinction between our synchronous coordination game and commitment games is that only a subset of actions are pledgeable in the synchronous setting, which constrains the set of implementable outcomes compared to the broader action sets possible in commitment games.

Our analysis also extends the literature on network coordination games. \citet{chwe1999structure}, \citet{berninghaus2002conventions}, \citet{demange2004group}, \citet{papadimitriou2021public}, \citet{leister2022social}, \citet{sadler2022ordinal}, and \citet{bayer2023best} have all contributed to understanding coordination in networked settings with strategic complementarities. The most relevant work is \citet{chwe2000communication}, which analyzes binary coordination games with common interests and strategic complementarities, providing insights into the structure of minimal communication networks required for efficient coordination. Our results build on these insights by extending the analysis to dynamic coordination settings and exploring the implications of interdependent payoffs.

The concept of strategic complementarities is central to our approach. \citet{milgrom1994monotone} provided foundational insights into comparative statics in games with strategic complementarities. Further technical results about lattice structures and equilibrium properties can be found in \citet{vives1990nash}, \citet{zhou1994set}, \citet{kramarz1996dynamic}, and \citet{topkis1998supermodularity}. \citet{echenique2004extensive} established a unified framework for analyzing dynamic strategic complementarities, though the restrictive conditions in that framework do not fully apply to our setting. We also draw on the literature on Markov perfect equilibria (MPE), particularly the work of \citet{curtat1996markov} and \citet{maskin2001markov}, which highlights the challenges of applying MPE to settings with irreversible actions.

Compared to the existing literature, our contribution lies in adding structure to the stage game through strategic complementarities, common interests, and deviation-proof conditions. These additional structures allow us to precisely determine local behavior and evaluate performance through a robust approach that focuses on the worst-case, least (monotone) equilibrium outcome. This adversarial approach is common in the study of coordination games, as seen in recent work by \citet{basak2022panics} and \citet{inostroza2023adversarial}, although the specific frameworks differ from ours.

The rest of the paper is organized as follows.  \autoref{sec:2} introduces the model setup. \autoref{sec:3} states the characterization of the synchronous game. \autoref{sec:asyn} states the characterization of the asynchronous game. \autoref{sec:equiv} builds the equivalence between the synchronous and asynchronous games. \autoref{sec:5} discusses several applications and extensions. \autoref{sec:colc} concludes.

\section{Model Setup} \label{sec:2}

\subsection{Stage Game}
A stage game is a normal form game that can be represented by $\Gamma(N,A,U)$: 

Players: $N:i=1,2,...,|N|$; Actions: $\forall i\in N,A_i=\{0,1\}$, $A:=\prod_{i\in N}A_i$, $ A_{-i}:=\prod_{j\in N\backslash i}A_j$; Utility: $U=(u_i)_{i\in N}: A\rightarrow \mathbb{R}^{|N|}$. 

The stage game is a normal form game, in which each player $i$ chooses an action $a_i\in A_i$ simultaneously\footnote{More precisely, the player could choose $s_i \in \Delta(A_i)$, which represents a mixed strategy. In this paper, we follow the standard approach on games of strategic complementarities and focus on pure strategy, see \cite{echenique2003mixed} and \cite{echenique2007finding}.}, and then the payoff is realized. 

We see $A_i$ as a lattice equipped with the order $1>0$; $\forall X\in 2^N$, we see $A_X:=\prod_{i\in X}A_i$ as a lattice equipped with the product order, i.e., $\forall a,b\in A_X$, $a>b$ if $\forall i\in X, a_i\geq b_i$, and $\exists i\in X, a_i>b_i$.

\begin{assumption}\label{ass:1}
\begin{enumerate}
    \item Single-crossing condition: $\forall i\in N, \forall a_{-i},a_{-i}'\in A_{N\backslash i}$ such that $a_{-i}'>a_{-i}$, $u_i(a_i=1,a_{-i})\geq u_i(a_i=0,a_{-i}) \Rightarrow u_i(a_i=1,a_{-i}')\geq u_i(a_i=0,a_{-i}')$, $u_i(a_i=1,a_{-i})>u_i(a_i=0,a_{-i}) \Rightarrow u_i(a_i=1,a_{-i}')>u_i(a_i=0,a_{-i}')$.
    \item Common interests\footnote{I am indebted to Tat-how Teh for suggesting this weaker condition to replace the positive spillovers in the original paper: $\forall i\in N, u_i(a_i, a_{-i})$ is non-decreasing in $a_{-i}$.}: $\forall i\in N, u_i^*(a_{-i}):=\max_{a_i\in A_i} u_i(a_i, a_{-i})$ is non-decreasing in $a_{-i}$, and $\forall a_{-i}'>a_{-i}$, if $u_i^*(a_{-i}')=u_i(a_i=1, a_{-i}'), u_i^*(a_{-i})=u_i(a_i=0, a_{-i})$, then $f(a_{-i}')>f(a_{-i})$\footnote{The tie-break rule is a technical condition that guarantees the existence of least equilibrium outcome.}.
    \item Deviation-proof condition: $\forall i\in N, \forall a_{-i}'>a_{-i}$, $u_i(a_i'=1,a_{-i}') \geq u_i(a_i=0,a_{-i}) \Rightarrow u_i(a_i'=1,a_{-i}') \geq u_i(a_i=0,a_{-i}')$, $u_i(a_i'=1,a_{-i}') > u_i(a_i=0,a_{-i}) \Rightarrow u_i(a_i'=1,a_{-i}') > u_i(a_i=0,a_{-i}')$.

\end{enumerate}
\end{assumption}

We briefly explain why we need them for the main characterization. The single-crossing condition implies that the stage game is of strategic complementarities, therefore, the best correspondence is non-decreasing. Consequently, the pure strategy Nash equilibria(NE) set is a non-empty complete lattice (see \citet{vives1990nash} or \citet{zhou1994set}). The common interests condition further implies that given the NE is ordered, it is compatible with the Pareto rank: the greater the equilibrium is, the better it is for all players, i.e., given any two pure strategy NE $a<a'$, $\forall i\in N, u_i(a)\leq u_i(a')$. The deviation-proof condition may seem weird at first which needs some additional justifications. In short, it guarantees that any of the sub-games is a coordination game, a detailed discussion is given in the appendix\footnote{We show in the appendix that no assumption cannot be removed with various counterexamples}. Last, we argue that these assumptions are fairly lenient. For instance, if $u_i(a_i=0,a_{-i})\equiv 0$, the deviation-proof condition is satisfied automatically, and therefore \autoref{ass:1} is equivalent to the requisite that: $u_i(a_i=1,a_{-i})$ is non-decreasing in $a_{-i}$ and $\nexists i\in N, a_{-i}\in A_{-i}$ such that $u_i(a_i=1,a_{-i})=0$.

\subsection{Synchronous Game}
A synchronous game $\Gamma(N,A,U,T)$ can be generated from the stage game $\Gamma(N,A,U)$ with additional structure:

Time: $t=1,2,...,T$; Players: $N:i=1,2,...,|N|$. Actions: $\forall i\in N, 1\leq t \leq T, A_{it}\equiv A_i=\{0,1\}$, $A_{t}:=\prod_{i\in N} A_{it}$, $A:=\prod_{i\in N}\prod_{1\leq t\leq T}A_{it}$. Utility: $U=(u_i)_{i\in N}: A_{T}\rightarrow \mathbb{R}^{|N|}$. $u_i$ satisfies the assumptions of the stage game. Irreversible constraint: $\forall i\in N, \forall 1\leq t\leq T-1$, $a_{it+1}\geq a_{it}$\footnote{Therefore, the synchronous game is also called monotone game\citep{gale2001monotone}.}. 

History: $\forall i\in N, H_{i1}\equiv H_1=\emptyset, \forall 2\leq t\leq T, H_{it}\equiv H_t =\prod_{1\leq k\leq t-1} A_{k}$, $\mathcal{H}_i\equiv \mathcal{H} =\bigcup_{1\leq t\leq T} H_{t}$, $\mathcal{A}_i=\bigcup_{1\leq t\leq T} A_{it}$. Pure strategy: $s=(s_i)_{i\in N}, \forall i\in N,s_i: \mathcal{H}_{i}\rightarrow \mathcal{A}_{i}$ such that the irreversible constraint is satisfied. 

The monotone game $\Gamma(N,A,U,T)$ is an extensive-form game with complete and almost perfect information: at the beginning of $ 1\leq t \leq T$, each player $i$ chooses an action $a_{it}\in A_{it}$ simultaneously; at the end of $t$, each player $i$ observes the history $h_{t}=\prod_{i\in N}\prod_{1\leq k\leq t-1}a_{ik} \in H_{t}$ perfectly; at the end of $t=T$, the payoff is realized. When the stage game is obvious, we write the monotone game $\Gamma(N,A,U,T)$ as $\Gamma(T)$ for short. In this paper, we restrict our attention to pure strategy to make the result elegant. We use $a_s\in A_{T}$ to represent the outcome of a pure strategy profile $s$.

One may explain the game as a limited commitment game with $T$ stages\footnote{The commitment is limited in the sense that only socially efficient action $1$ is pledgeable. This assumption can be relaxed, see the appendix for a more detailed discussion.}: $T-1$ commitment-making stages and one decision-making stage, which explains why the payoff depends on the action in the last stage only. The commitment is creditable and irreversible, which means once a player pledges to take the socially plausible action $1$, he can neither withdraw nor violate his commitment. 

\subsection{Asynchronous Game}

An asynchronous game $\Gamma(N,A,U,T,\{N_t\}_{t=1}^{T})$ can be generated from the stage game $\Gamma(N,A,U)$ with additional structure:

Time: $t=1,2,...,T$. Partition: $\forall p\neq q, 1\leq p,q\leq T, N_p\cap N_q =\emptyset, \bigcup_{1\leq t\leq T}N_t=N$.

History: $H_{1}:=\emptyset$, $\forall 2\leq t\leq T, \forall i\in N_t, H_{it}=H_{t}:=\prod_{1\leq k\leq t-1}\prod_{i\in N_k}A_{i}$. Pure strategy: $s=(s_i)_{i\in N}$, $\forall 1\leq t\leq T, \forall i\in N_t$, $s_i: H_{t}\rightarrow A_i$.

The asynchronous game $\Gamma(N,A,U,T,\{N_t\}_{t=1}^{T})$ is an extensive-form game with complete and almost perfect information: at the beginning of $ 1\leq t \leq T$, each player $i\in N_t$ chooses an action $a_{i}\in A_{i}$ simultaneously; at the end of $t$, each player $j\in N$ observes the history $h_{t}=\prod_{1\leq k\leq t}\prod_{i\in N_t}a_{i} \in H_{t}$ perfectly; at the end of $t=T$, the payoff is realized.  In this paper, we restrict our attention to pure strategy to make the result elegant\footnote{Most characterization maintains even we consider mixed strategy.}. We use $a_s\in A$ to represent the outcome of a pure strategy profile $s$\footnote{$a_s$ can be generated in the following way: $\forall i\in N_1, a_{i}=s_{i}$; $\forall  2\leq k \leq T, i\in N_k, a_i=s_i(\prod_{1\leq l\leq k-1}a_{N_l})$}.

\subsection{Solution Concept}\label{sec:sou}

Both synchronous and asynchronous games are extensive-form games with complete and almost perfect information, so the usual solution concept is sub-game perfect equilibrium(SPNE). However, in this paper, we will introduce a new solution concept: Monotonic sub-game perfect equilibrium(MSPNE), the SPNE with monotonic strategy. 

\begin{definition}[Monotonic Strategy]
A pure strategy $s$ is monotonic in history if $\forall i\in N, \forall 2\leq t\leq T,\forall h_1,h_2\in H_{it}$, $h_1>h_2$ implies $s_i(h_1)\geq s_i(h_2)$. A pure strategy profile is a MSPNE if it is both a SPNE and monotonic in history. 
\end{definition}

The monotonic strategy captures the feature that the strategy maps higher history into higher action. For a monotone game whose stage game is of strategic complementarities and common interests, it seems to be a desirable property\footnote{A more theoretical explanation is that: we want a monotone comparative statics stronger than the strong set order.}. For more discussions that justify the reasonability of MSPNE, see \autoref{sec:counterexample}. In this paper, we will focus on MSPNE. 

We introduce some simplified notations for convenience. Given $a$ is a strategy profile, we may also use $X_a:=\{i\in N: a_i=1\}$ to represent it equivalently.  We use $O(\Gamma(N,A,U))$ to represent the pure strategy Nash equilibria set; when the stage game is obvious, we also write it as $O(\Gamma)$ for short. We use $O(\Gamma(T))$ to represent pure strategy MSPNE outcome of the synchronous game, and $O(\Gamma(T,\{N_t\}_{t=1}^{T}))$ to represent pure strategy MSPNE outcome of the asynchronous game.

\subsection{Graph Preliminaries}
We will also introduce some preliminary graph basics.

\begin{definition}[Directed Graph]
A directed graph $G=(N,E)$, in which $E\in 2^{N\times N}$ contains some ordered pairs. $\forall i,j\in N, (i,j)\in E$ represents $i\rightarrow j$. For convenience, we assume that $(i,i)\notin E$. $E_i:=\{j\in N: (j,i)\in E\}$. Let $\mathcal{G}$ represent the class of directed graph. Given a directed graph $G$, we define the subgraph $G|_{X}: 2^{N}\rightarrow \mathcal{G}$, $G|_{X}=(X,E\cap (X\times X))$.
\end{definition}

\begin{definition}[Strongly Connected]
Given a directed graph $G$, $R: 2^N\rightarrow 2^N$, $R(X)=\{j\in N: \exists \{j_p\}_{p=1}^{q}, j_1=j,j_q\in X, \forall 1\leq p\leq q-1, (j_p,j_p+1)\in E\}$. $R$ represents all of the vertex that can reach some vertex in $X$ according to $G$. A graph $G$ is strongly connected if $\forall i\in N, R(\{i\})=N$. Otherwise, a directed graph can be divided into several maximal strongly connected subgraphs: each subgraph is strongly connected, and there exists no larger subgraph that is strongly connected. Such a maximal strongly connected subgraph is called a strongly connected component\footnote{For convenience, we see a singleton as strongly connected.}. 
\end{definition}

 \begin{definition}[Tree-depth]
Tree-depth(td) is recursively defined: $td: \mathcal{G} \rightarrow \mathbb{N}^+$\footnote{At the best of my knowledge, tree-depth is defined for the undirected graph only \citep{pothen1988complexity}, and we extend the definition to the directed graph.}.
 
\begin{equation*}	   td(G):=
 		\begin{cases}

    1 & \text{if } |G|=0\\
1+\text{min}_{i\in N }td(G\backslash i) & \text{else if G is strongly connected}\\
\text{max}_{k} td(G_k) & \text{else, where $G_k$ is the strongly connected component}
 		\end{cases}
 	\end{equation*}  

Given a directed graph $G$, we also define subgraph tree-depth(std): $std: 2^{N} \rightarrow \mathbb{N}^+$. $std(X;G)=td(G|_{R(X)})$.
 \end{definition}

\begin{definition}[Weakest Link Game]
Weakest Link Game $\Lambda(G(N,E),A,U)$ is a stage game: Players: $N$; Actions: $A=\{ 0,1\}^{N}$; Directed graph $G=(N,E)$, Payoff: $U=(u_i)_{i\in N}: A\rightarrow \mathbb{R}^{|N|}$ satisfies \autoref{ass:1}, and additionally $\forall i\in N, \forall a_{-i}\in A_{-i}$, $u_i(a_i=1,a_{-i})>u_i(a_i=0,a_{-i})$ if $\forall j\in E_i, a_j=1$, and $u_i(a_i=1,a_{-i})<u_i(a_i=0,a_{-i})$ otherwise\footnote{Alternatively, we could assume $u_i(a_i=1,a_{-i})\leq u_i(a_i=0,a_{-i})$ otherwise. However, the tie-break rule we choose is convenient for the consideration of the least equilibrium outcome}.
\end{definition}

For instance, the minimum effort game on a directed graph is one of the weakest link game: $\forall i\in N, u_i(a)=a_i\times [\text{inf}_{j\in E_i}a_j-\frac{1}{2}]$. The pure strategy NE of the strict weakest link game on a directed graph is simple: $a$ is a NE if and only if $\forall i\in N, a_i=\text{inf}_{j\in E_i}a_j$.

\section{Synchronous Game: Characterization}\label{sec:3}
\subsection{Theoretical Characterization}

The monotone game is an extensive-form game with complete and almost perfect information, the existence of mixed strategy SPNE is well-established. However, the existence of MSPNE in pure strategy is more subtle. We will show the existence first.

In this paper, we allow the possibility that some players have (iterated) strictly dominant strategy in the stage game, and we show that those players will take the dominant strategy for sure in all SPNE, and therefore the analysis can be focused on the case that nobody has (weakly) dominant strategy without loss of generality. The case that some players have a strictly dominant strategy $1$ is easy, while the case that some players have a strictly dominant strategy $0$ is more subtle.

\begin{lemma}\label{lem:0}
Every SPNE outcome of the synchronous game is an action profile of the stage game that survives the iterated elimination of strictly dominated strategy.

\end{lemma}

The above argument shows that we can without loss of generality focus on the non-degenerate case. Therefore, in the following, we assume that nobody has a dominant action in the stage game\footnote{We also eliminate the possibility of weak dominance for elegance}. 

\begin{assumption}
$\forall i\in N$, $u_i(a_i'=1,a_{-i}=1^{N\backslash i})>u_i(a_i=0,a_{-i}=1^{N\backslash i})$, $u_i(a_i=0,a_{-i}=0^{N\backslash i})> u_i(a_i'=1,a_{-i}=1^{N\backslash i})$.
\end{assumption}

The existence of pure strategy MSPNE is clear now: $\forall i\in N, \forall h\in \mathcal{H}, s_i(h)\equiv 1$ is a MSPNE.

The first main result we want to show is: the pure strategy MSPNE outcome shrinks as there are more commitment stages. To prove this, we need a lemma.

\begin{lemma}\label{lem:1}
Suppose $s$ is a pure strategy MSPNE(SPNE), given $b_k\in A_k, 1\leq k\leq T-1$ satisfies the monotone constraint, define $a_1:=s_1$, $a_2:=s_2(a_1\vee b_1)$, $a_k:=s_k(\prod_{t=1}^{k-1} (a_t\vee b_t))$,  the following strategy $g$ is a MSPNE(SPNE) that realizes the same outcome: $\forall i\in N, g_{i1}=0, \forall 2\leq t\leq T-1, g_{it}(\prod_{k=1}^{t-1}b_k)=1 \text{ iff } b_{it-1}=1, g_T(\prod_{k=1}^{T-1}b_k)=s_T(\prod_{t=1}^{T-1} (a_t\vee b_t))$. 

\end{lemma}

The detailed proof is left in the appendix. But we will see some intuitions on why this lemma is correct: $\{a_t\}_{t=1}^{T-1}$ are the actions that players should take according to strategy $s$, while $\{b_t\}_{t=1}^{T-1}$ are the actions that players do take. The strategy $g$  is constructed through history transformation: if the action you take is lower than the action you should take, all players act as if you have taken the higher action since the payoff depends on the action at the terminal only; else if the action you take is higher than the action you should take, given the higher action is irreversible, all players have to act based on the action you take. 

This implies any of the MSPNE outcomes is necessarily self-enforcing: there exists no MSPNE outcome that can be implemented via pledging only. If all players agree on a MSPNE outcome, then commitment is not necessary. This does not hold for general commitment games, in which the deviation-proof assumption may fail to hold. Now we can show that we can neglect those players who have a (iterated) strictly dominant strategy $0$.

Now we can show that the outcome set of MSPNE shrinks weakly. The validness of the proposition is guaranteed by the above lemma: the commitment stages are non-informative for analyzing MSPNE, so we could remove it without loss of generality.

\begin{proposition}\label{prop:1}
$\forall T\geq 2, O(\Gamma(T+1))\subset O(\Gamma(T)) \subset O(\Gamma)$.
\end{proposition}

We also prove a result similar to \citet{farrell1985standardization} and \citet{gale1995dynamic} as a direct corollary of \autoref{lem:1}, which shows that the unique MSPNE outcome of the monotone game is the socially efficient one, given the commitment stages are long enough.

\begin{corollary}\label{cor:1}
$O(\Gamma(|N|))=\{N\}$.
\end{corollary}

\begin{proof}
Using mathematical induction. The argument is true for any of the monotone games with $|N|=2$. Suppose the argument is true for any of the monotone games with $|N|=k$. Then suppose it fails for a monotone game with $|N|=k+1$, therefore at least one player $i$ takes action $0$ as the outcome of the strategy profile. We only consider the strategy in the form of \autoref{lem:1}, and we claim that such player $i$ has the incentive to deviate at $t=1$: in the remaining sub-game after deviation, mathematical induction implies in all MSPNE of the sub-game, all players take $1$, which is therefore a profitable deviation.
\end{proof}

In fact, for an adversarial coordination approach only, we could immediately focus on the characterization of the minimum number of stages needed to guarantee all players take the socially efficient action in all MSPNE, i.e., the smallest $T\in \mathbb{N}^+$ such that $O(\Gamma(T))=\{N\}$, which is well-defined. Nevertheless, we will show a general picture of the MSPNE outcome set.

\begin{figure}[h]
  \centering
  \begin{tikzpicture}[>=stealth, node distance=2cm]
    
    \coordinate (1) at (1,1);
    \coordinate (2) at (3,3);
    \coordinate (6) at (4,4);
    \coordinate (3) at (4,5);
    \coordinate (4) at (5,4);
    \coordinate (5) at (6,6);

    \filldraw[black] (1) circle (2pt);
    \filldraw[black] (2) circle (2pt);
    \filldraw[black] (6) circle (2pt);
    \filldraw[black] (3) circle (2pt);
    \filldraw[black] (4) circle (2pt);
    \filldraw[black] (5) circle (2pt);

    \draw[dashed, red, rotate around={45:(3.5,3.5)}] (3.5,3.5) ellipse (4.5cm and 2.5cm) node[right=0cm, above=-4.0cm] {$O(\Gamma)$};
    \draw[dashed, red, rotate around={45:(5.0,5.0)}] (5.0,5.0) ellipse (2.0cm and 1.5cm) node[right=0cm, above=-2.5cm] {$O(\Gamma(2))$};
    \draw[dashed, red] (5) circle (0.4) node[right=0cm, above=-1.0cm] {$O(\Gamma(3))$};
  \end{tikzpicture}
  \caption{A Graphical Illustration of MSPNE lattice}
  \label{fig:x}
\end{figure}
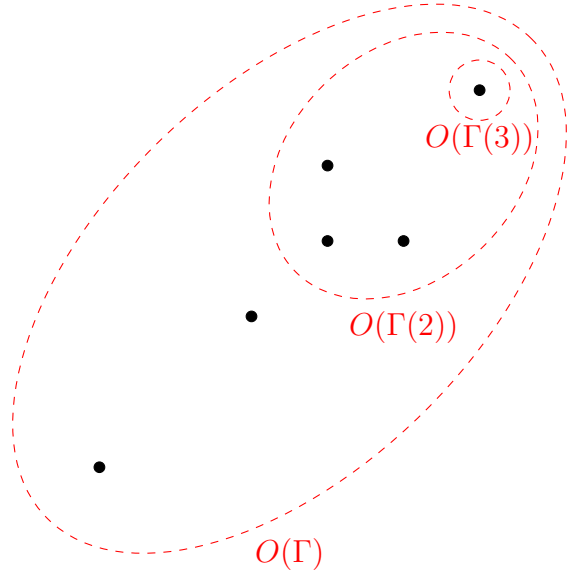

The following lemma is a preliminary tool to prove that the outcome set of MSPNE is a lattice, as it confirms the existence of minimum elements. The detailed proof is left in the appendix, and the key is to construct the most conservative strategy, that is, taking the lowest possible action collectively.

\begin{lemma}\label{lem:2}
If $\forall i\in N, \exists N_i\in O(\Gamma(T)) \text{ such that } i\notin N_i$, then $\emptyset \in O(\Gamma(T))$.  
\end{lemma}

Now we define two operators: $\phi_{\Gamma}$ and $\tau_{\Gamma}$, which play a crucial role in characterizing MSPNE outcome.

\begin{definition}\label{def:3}
$\phi_{\Gamma}: \mathbb{N}^+ \rightarrow 2^N, \phi_{\Gamma}(T):= \text{min} \{X: \forall Y\in O(\Gamma(T)), X\subset Y\}$ which represents those players who take action $1$ in all MSPNE of $\Gamma(T)$. $\tau_{\Gamma}: 2^N\rightarrow  \mathbb{N}^+, \tau_{\Gamma}(X):=\text{min} \{T: \forall Y\in O(\Gamma(T)), X \subset Y\}$, which represents the smallest $T$ such that those players $X$ take action $1$ in all MSPNE of $\Gamma(T)$.
\end{definition}

Note that \autoref{prop:1} and \autoref{cor:1} guarantees that $\tau_{\Gamma}$ is well-defined, and both $\phi_{\Gamma}$ and $\tau_{\Gamma}$ are bounded, non-decreasing single-valued correspondence. One may further see these two operators as generalized converse, though both operators are neither injective nor subjective in general: $\phi_{\Gamma}(T)=\text{max} \{X\in 2^N: \tau_{\Gamma}(X)\leq T\}$, $\tau_{\Gamma}(X)=\text{min} \{T\in \mathbb{N}^+: \phi_{\Gamma}(T)\supset X \}$. \autoref{lem:2} further shows $\phi_{\Gamma}(T)$ is the least element of $O(\Gamma(T))$. 

We also need two specific auxiliary games. The players in the former game act as if those players outside of the auxiliary game all take the highest action for sure, while the players in the latter act as if those players outside of the auxiliary game all take the lowest action for sure. 

\begin{definition}
We define $\overline{\Gamma}(X)$ as the following stage game: $(X,A|_{X},\overline{U}|_{X})$, where $\forall i\in X, \forall Y\subset X$, $\overline{u}_i(Y)=u_i(Y\cup (N\backslash X))$. We define $\underline{\Gamma}(X)$ as the following stage game: $(X,A|_{X},\underline{U}|_{X})$, where $\forall i\in X, \forall Y\subset X$, $\underline{u}_i(Y)=u_i(Y)$.
\end{definition}

Now we are well prepared to give a complete characterization of $O(\Gamma(T))$.

\begin{proposition}\label{prop:2}
$X$ is a MSPNE outcome of $\Gamma(T)$ iff: $X$ is a NE and $\phi_{\overline{\Gamma}(N\backslash X)}(T)=\emptyset$.
\end{proposition}

The sketch of the proof: On the one hand, given $X$ is a MSPNE, \autoref{lem:1} makes sure there exists a MSPNE such that all players do not pledge in at all and those players in $X$ take action $1$ in the decision-making stage. Given the strategy is monotonic in history, those players in $X$ take action $1$ regardless of the history, so those players outside of $X$ may simplify the game into $\overline{\Gamma}(N\backslash X)(T)$ and apply \autoref{lem:2} now. On the other hand, one may construct a MSPNE in the form of \autoref{lem:1} via the most conservative strategy in the proof of \autoref{lem:2}, which realizes the outcome $X$.

\begin{corollary}
\begin{enumerate}
    \item $O(\Gamma(T))$ is a non-empty complete lattice: $\forall a,b\in O(\Gamma(T))$, $a\wedge^{\Gamma(T)} b=a\wedge^{\Gamma} b$, $a\vee^{\Gamma(T)} b= (a\vee^{\Gamma} b) \cup \phi_{\overline{\Gamma}(N\backslash (a\vee^{\Gamma} b))}(T)$.
    \item $O(\Gamma(T))$ is weakly increasing in weak set order: $\text{sup } O(\Gamma(T+1)) \geq \text{sup } O(\Gamma(T))$, \text{inf} $O(\Gamma(T+1)) \geq \text{inf } O(\Gamma(T))$.
    \item $O(\Gamma(T+1))$ is not a sub-lattice of $O(\Gamma(T))$ in general.
\end{enumerate}
\end{corollary}

\begin{proof}
Given $O(\Gamma(T))\subset O(\Gamma)$, we must have $\forall a,b\in O(\Gamma(T)), a\wedge^{\Gamma(T)} b \leq a\wedge^{\Gamma} b$, so it is enough to show that $a\wedge^{\Gamma} b \in O(\Gamma(T))$: $\phi_{\overline{\Gamma}(N\backslash (a\wedge^{\Gamma} b))}(T)\subset \phi_{\overline{\Gamma}(N\backslash a)}(T) \cup \phi_{\overline{\Gamma}(N\backslash b)}(T)=\emptyset$. Given $O(\Gamma(T))\subset O(\Gamma)$, we must have $\forall a,b\in O(\Gamma(T)), a\vee^{\Gamma(T)} b \geq a\vee^{\Gamma} b$, $a\vee^{\Gamma(T)} b$ is the least MSPNE that is greater than  $a\vee^{\Gamma} b$. Weak set order is obvious since $\phi_{\Gamma}(T)$ is weakly increasing. One counterexample is given in the following: consider the following weakest link game, not hard to show that, $O(\Gamma)=\{ \emptyset, \{1,2,3 \}, \{4,5,6 \}, \{1,2,3,4,5,6 \}, N \}$, while $O(\Gamma(T=2))=\{ \emptyset, \{1,2,3 \}, \{4,5,6 \}, N \}$.

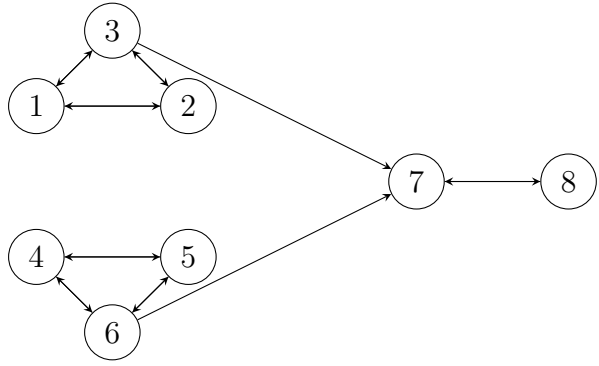
\begin{figure}[h]
  \centering
\begin{tikzpicture}[>=stealth, node distance=2.5cm, every node/.style={circle, draw, minimum size=20pt}]
    
    \node (1) at (-1,2) {1};
    \node (2) at (1,2) {2};
    \node (3) at (0,3) {3};
    \node (4) at (-1,0) {4};
    \node (5) at (1,0) {5};
    \node (6) at (0,-1) {6};
    \node (7) at (4,1) {7};
    \node (8) at (6,1) {8};

    \draw[->] (1) -- (2);
    \draw[->] (1) -- (3);
    \draw[->] (2) -- (1);
    \draw[->] (2) -- (3);
    \draw[->] (3) -- (1);
    \draw[->] (3) -- (2);
    
    \draw[->] (4) -- (5);
    \draw[->] (4) -- (6);
    \draw[->] (5) -- (4);
    \draw[->] (5) -- (6);
    \draw[->] (6) -- (4);
    \draw[->] (6) -- (5);
    
    \draw[->] (3) -- (7);
    \draw[->] (6) -- (7);
    \draw[->] (8) -- (7);
    \draw[->] (7) -- (8);
  \end{tikzpicture}
  \caption{$O(\Gamma(T+1))$ is possibly not a sub-lattice of $O(\Gamma(T))$}
  \label{fig:5}
\end{figure}

\end{proof}

\subsection{Recursive Characterization}

As we haven shown in \autoref{prop:2}, $\phi_{\Gamma}$, or $\tau_{\Gamma}$ uniquely determines the MSPNE of the monotone game. We specifically take care of the smallest $T\in \mathbb{N}^+$ such that $\phi_{\Gamma}(T)=N$. \autoref{cor:1} derives a bound that such $T\leq |N|$. Let us see some intuition about why the bound above is possibly not tight first.

\begin{figure}[h]
  \centering
  \begin{tikzpicture}[>=stealth, node distance=2cm, every node/.style={circle, draw, minimum size=20pt}]
    
    \node (1) at (0,0) {1};
    \node (2) at (30:2) {2};
    \node (3) at (90:2) {3};
    \node (4) at (150:2) {4};
    \node (5) at (210:2) {5};
    \node (6) at (270:2) {6};
    \node (7) at (330:2) {7};

    \foreach \x in {2,...,7} {
        \draw[->] (1) -- (\x);
        \draw[->] (\x) -- (1);
    }
  \end{tikzpicture}
  \caption{Why SPNE is fast: Weakest Link Game with Star Network}
  \label{fig:3}
\end{figure}
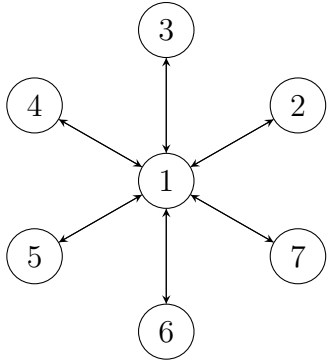

The first intuitive example is the weakest link game with a star network. In this example, only two stages are needed, since the central players could always pledge in the commitment stage to guarantee the efficient outcome.

\begin{figure}[h]
  \centering
  \begin{tikzpicture}[>=stealth, node distance=2cm, every node/.style={circle, draw, minimum size=20pt}]
    
    \node (1) at (90:2) {1};
    \node (2) at (45:2) {2};
    \node (3) at (0:2) {3};
    \node (4) at (-45:2) {4};
    \node (5) at (-90:2) {5};
    \node (6) at (-135:2) {6};
    \node (7) at (180:2) {7};
    \node (8) at (135:2) {8};

    \foreach \x/\y in {1/2, 2/3, 3/4, 4/5, 5/6, 6/7, 7/8, 8/1} {
        \draw[->] (\x) -- (\y);
    }
  \end{tikzpicture}
  \caption{Why SPNE is fast: Weakest Link Game with Chain Network}
  \label{fig:2}
\end{figure}
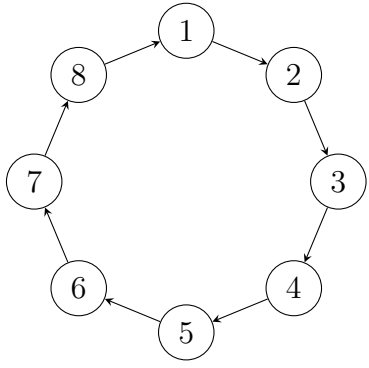

The second intuitive example is the weakest link game with a chain network. In this example, only two stages are needed, since any of the players could always pledge in the commitment stage, then the direct successor has a strictly dominant strategy, and the successor of the successor has an iterated strictly dominant strategy, and so on, which induces the efficient outcome.

\begin{figure}[h]
  \centering
\begin{tikzpicture}[>=stealth, node distance=2cm, every node/.style={circle, draw, minimum size=20pt}]
    
    \node (1) at (0,0) {1};
    \node (2) at (3,0) {2};
    \node (3) at (3,-3) {3};
    \node (4) at (0,-3) {4};

    \draw[->] (1) -- (2);
    \draw[->] (1) -- (4);
    \draw[->] (2) -- (1);
    \draw[->] (2) -- (3);
    \draw[->] (3) -- (4);
    \draw[->] (4) -- (3);
  \end{tikzpicture}
  \caption{Why MSPNE is faster}
  \label{fig:6}
\end{figure}
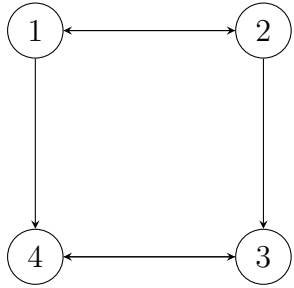

The last intuitive example is simply the weakest link game with several strongly connected components, which has been discussed in detail. As we have shown, MSPNE admits a simple analysis: once we have shown that a sub-group of players take action $1$ in all MSPNE, other players could without loss of generality simplify the game by removing the sub-group of players.

These three examples count three different mechanisms, which count all cases that facilitate adversarial coordination. Though one cannot compute $\phi_{\Gamma}$ directly in general, a recursive algorithm is developed to determine the generalized inverse $\tau_{\Gamma}$. We still need a concept similar to strict equilibrium, in the sense that all those players who take action $1$ prefer to do so strictly.

\begin{definition}[Strict Sufficiency]
Strictly sufficient set(SSS) is a subgroup of players who strictly prefer taking action $1$ collectively: $SSS(\Gamma)=\{X\in 2^N: \forall i\in X, u_i(X)>u_i(X\backslash i)\}$. Strictly sufficient equilibria is both a strictly sufficient set and a Nash equilibrium: $SSE(\Gamma)=\{ X\in 2^N: X\in SSS(\Gamma) \bigcap O(\Gamma)\}$.
\end{definition}

Both $SSS$ and $SSE$ are non-empty since $N\in SSE$. We also define $\emptyset \notin SSS$. Now we show the main theorem of this paper.

\begin{theorem}\label{thm:1}
$\tau_{\Gamma}(X)=\text{min}_{Y\in SSS(\Gamma): Y\supset X} \tau(\underline{\Gamma}(Y))$.
\begin{footnotesize}
   \begin{equation*}	    
\tau(\Gamma)= \begin{cases}
 \tau(\overline{\Gamma}(N\backslash i)) & \text{if } u_i(i)>u_i(\emptyset)\\
    1 & \text{else if } N=\emptyset\\
    \text{min} \{ 1+ \text{min}_{i\in N} \tau(\overline{\Gamma}(N\backslash i)), min_{X\in SSS(\Gamma)\backslash \{\emptyset, N\} } max\{ \tau(\underline{\Gamma}(X)),\tau(\overline{\Gamma}(N\backslash X)) \}\} & \text{else } 
 		\end{cases}
 	\end{equation*} 
\end{footnotesize}
All results hold if we change $SSE$ into $SSS$.
\end{theorem}
 
The detailed proof is long and therefore left in the appendix. But the intuition of the result is fairly clear. For the recursive algorithm, the first case is degenerated if some players have a strictly dominant strategy, and the second case is simply left for the counting of the first case. The last case is more complex: first, the minimum number of stages needed to make all those players take action $1$ is no more than the right-hand-side, which is an upper bound: otherwise, either some player $i$ has the incentive to deviate, or some sub-group of players have the incentive to deviate. The upper bound is further tight. To see this, $t_1:=\text{min}\{t\in \mathbb{N}^+: \phi_{\Gamma}(t)\neq \emptyset \}$, we only need to consider the case that $\phi_{\Gamma}(t_1)=N$ (otherwise we can choose $X=\phi_{\Gamma}(t_1)$, which meets the right hand side then). In this case, we could show that if $T<t_1$, then $\emptyset \in O(\Gamma(T))$ by constructing the most conservative strategy, which completes the proof of the recursive algorithm. The first statement just makes use of the definition of $\phi_{\Gamma}(T)$, and the last statement holds since the equality holds in SSE.

In practice, SSE is significantly fewer than SSS, so it is better to check SSE; while it is usually more convenient to write proof with SSS rather than SSE since SSS only shrinks or expands with respect to most monotone comparative statics.

There is still something worth saying in the theorem, as it will be used in the later proof: We may see $\tau(\Gamma)$ as the value function and define the policy function consists of the following three operations: dominate ($i$), delete ($i$), and divide (the game into two auxiliary games $\underline{\Gamma}(X)$ and $\overline{\Gamma}(N\backslash X)$). 

Now we develop the corresponding monotone comparative statics like \cite{milgrom1994monotone}. Instead of using an algebraic method, we provide a direct proof through the recursive algorithm.
\begin{proposition}\label{prop:3}
Let $\{\Gamma_\theta(T)=(N,A,U_\theta=\prod_{i\in N}u_i(a_i,a_{-i},\theta),T)\}$ be a family of monotone games such that $\forall i\in N$, $u_i(a_i;a_{-i},\theta)$ further satisfies single-crossing condition in $(a_i;a_{-i},\theta)$, then the least MSPNE is the non-decreasing function of $\theta$.
\end{proposition}

\begin{proof}
We just need to show that $\forall T\in \mathbb{N}^+, \phi_{\Gamma_{\theta}}(T)$ is non-decreasing in $\theta$, which is equivalent to show that $\forall X\in 2^{N}, \tau_{\Gamma_{\theta}}(X)$ is non-increasing in $\theta$, which can be checked in the algorithm that, all feasible operations under $\theta$ are still feasible under $\theta'$ when $\theta'>\theta$: if $i$ is dominated under $\theta$, the same under $\theta'$; else if $X$ is a strictly sufficient set under $\theta$, the same under $\theta'$, which proves the desired result.
\end{proof}

\section{Asynchronous Game: Characterization}\label{sec:asyn}

The following lemma shows that we can restrict our attention to the subgame that survives iterated elimination of strictly dominated strategy without loss of generality.
\begin{lemma}\label{lemma1}
Any MSPNE outcome of the asynchronous coordination game is a NE of the stage game.
\end{lemma}

Therefore, in the following, we assume that no player has strictly dominated strategy.

\begin{assumption}
$\forall i\in N$, $u_i(a_i=1,a_{-i}=1^{N\backslash i})>u_i(a_i=0,a_{-i}=1^{N\backslash i})$, $u_i(a_i=0,a_{-i}=0^{N\backslash i})>u_i(a_i=1,a_{-i}=0^{N\backslash i})$.
\end{assumption}

The first theorem characterizes the least MSPNE. We introduce two definitions first.

\begin{definition}
For an asynchronous coordination game $\Gamma(N,A,U,T,\{N_t\}_{t=1}^{T})$, the action profile $a^*$, the least action profile that survives \textbf{iterated elimination of strictly extensively dominated strategy}(IESEDS) is defined recursively:

1. $\forall i\in N_T$, $\forall h_{T-1}\in H_{T-1}$, consider the auxiliary game $\Gamma(N_T,A_T,U'(h_{T-1}))$, in which $U'(h_{T-1}): A_T\rightarrow \mathbb{R}^{|N_T|}$. $ \forall i\in N_T, a_{T}\in A_{T}, u_i'(a_T)=u_i(a_{T},h_{T-1})$. Compute the least action profile that survives iterated elimination of strictly dominated strategy as $\hat{a_{T}}(h_{T-1})$ as $a_T^*(h_{T-1})$.

2. Suppose for some $1\leq k\leq T-1$, $a_{T-k+1}^*(h_{T-k})$ is well-defined for every $h_{T-k}\in H_{T-k}$. Now $\forall i\in N_{T-k}$, $\forall h_{T-k-1}\in H_{T-k-1}$, consider the following auxiliary game $\Gamma(N_{T-k},A_{T-k},U'(h_{T-k-1}))$, in which $U'(h_{T-k-1}): A_{T-k}\rightarrow \mathbb{R}^{|N_{T-k}|}$. $\forall i\in N_{T-k}, a_{T-k}\in A_{T-k}, u_i'(a_{T-k})=u_i(a_{T-k},a_{T-k+1}^*(h_{T-k-1},a_{T-k}), h_{T-k-1})$. Compute the least action profile that survives iterated elimination of strictly dominated strategy as $\hat{a_{T-k}}(h_{T-k-1})$. And $a_{T-k}^*((h_{T-k-1})):=(\hat{a_{T-k}}(h_{T-k-1}), a_{T-k+1}^*(h_{T-k-1},a_{T-k}))$.  

3. $a^*=(a_t^*((a_{k}^*)_{k=1}^{t-1}))_{t=1}^{T}$.

\end{definition}

\begin{definition}
For a stage game $\Gamma(N,A,U)$, a directed graph $G=(M,E), M\in 2^N$ is $M$-sufficient if: $\forall i\in M, u_i(E_i\cup i)>u_i(E_i)$. A directed graph $G=(M,E), M\in 2^N$ is $M$-minimal sufficient if it is $M$ sufficient and there exists no sufficient graph $G'=(M,E')$ such that $E'\subsetneq E$. A directed graph is sufficient(minimal sufficient) if $M=N$. The class of sufficient(minimal sufficient) graph of the stage game $\Gamma$ is $\mathcal{S}_{\Gamma}$($\mathcal{MS}_{\Gamma}$).

For an asynchronous coordination game $\Gamma(N,A,U,T,\{N_t\}_{t=1}^{T})$, a directed graph $G=(M,E), M\in 2^N$ is $M$-feasible if: it is $M$-sufficient, and $\forall i,j\in N_t\cap M, 1\leq t\leq T$, $i$ and $j$ are not strongly connected in $G|_{\cup_{t\leq k\leq T}N_k\cap M}$. A directed graph is feasible if $M=N$.
\end{definition}

For instance, when the stage game is the weakest link game $\Lambda(G(N,E),A,U)$, the unique minimal sufficient graph is the graph $G$. 

\begin{theorem}\label{theorem1}
For an asynchronous coordination game $\Gamma(N,A,U,T,\{N_t\}_{t=1}^{T})$, the following three statements are equivalent:

1. Among all MSPNE outcomes, all players in $M\in 2^N$ take action 1.

2. In the least action profile survives IESEDS, all players in $M\in 2^N$ take action 1.

3. There exists a $S$-consistent directed graph, $S\supset M$.
\end{theorem}

The theorem also shows the existence of the least MSPNE outcome for every asynchronous coordination game fixed. The following proposition further shows the existence of a partition that realizes the greatest least MSPNE outcome for every stage game fixed.

\begin{proposition}\label{proposition1}
Given any $T\in \mathbb{N}^+$ fixed, there exists a greatest subset of players $M(T)\subset N$ such that: there exists a partition such that among all MSPNE outcomes, all players in $M(T)\subset N$ take action 1. $M(T)\subset M(T+1)$. 
\end{proposition}

Now the corresponding problem of finding the optimal asynchronous partition for a fixed stage game $\Gamma(N,A,U)$ and the number of periods $T$ can be turned into the search for the greatest feasible graph. However, the definition of the feasible graph is relatively hard to check. The following lemma is useful: We use the Tree-depth of the sufficient graph to characterize the existence of feasible graph.

\begin{lemma}\label{lemma2}

For a directed graph $G=(N,E)$, we have the following equivalent characterization of Tree-depth:

1) If there exists a partition $\{N_t\}_{t=1}^{T}$ such that $\forall i,j\in N_t, 1\leq t\leq T$, $i$ and $j$ are not strongly connected in $G|_{\cup_{t\leq k\leq T}N_k}$, then we have $Td(G)\leq T$.

2) If $Td(G)\leq T$, then there exists a partition $\{N_t\}_{t=1}^{T}$ such that $\forall i,j\in N_t\cap S, 1\leq t\leq T$, $i$ and $j$ are not strongly connected in $G|_{\cup_{t\leq k\leq T}N_k}$. 
\end{lemma}

That is to say, it is enough to search for the greatest sufficient graph whose tree-depth is no more than $T$. In the following, we will further provide a surprising equivalence between the optimal asynchronous coordination design and the synchronous game.

\section{Equivalent Characterization}\label{sec:equiv}

We briefly describe the strategy we use to prove the equivalence: we will use the Tree-depth as a bridge. We first show that the synchronous game can be turned into the search for the greatest sufficient graph whose Tree-depth is no larger than $T$. Next, we show that a sufficient graph whose Tree-depth is no larger than $T$ is equivalent to a feasible graph with respect to some $T$-partition. Therefore, the synchronous game is equivalent to the optimal asynchronous design.

The first result is that: Tree-depth is exactly the solution to the synchronous game whose stage game is the weakest link game. That is, the corresponding $\tau_{\Lambda}$ has a much simpler form:

 \begin{proposition}\label{prop:4}
For the synchronous game whose stage game is the weakest link game, $\tau_{\Lambda}(X)=td(G|_{W(X)})$
\begin{equation*}	   td(G):=
 		\begin{cases}

    1 & \text{if } |G|=1\\
1+\text{min}_{i\in N }td(G\backslash i) & \text{else if G is strongly connected}\\
\text{max}_{i } td(G_i) & \text{else, where $G_i$ is the strongly connected component}
 		\end{cases}
 	\end{equation*}  
 \end{proposition}

The proof of \autoref{prop:4} is similar to \autoref{thm:1}, see appendix for details. 

Tree-depth is the solution to the optimal elimination trees problem, and the corresponding decision problem is NP-complete, see \cite{pothen1988complexity}. Therefore, it is in general NP-hard to determine whether $\emptyset \in O(\Gamma(T))$ for some $T$. 

Now we show how to turn a stage game into a weakest link game on a direct graph. The intuition is to apply monotone comparative statics and move down the payoff of each player to a threshold level. One can easily turn a stage game into a weakest link game using minimal sufficient network: pick any $S\in \mathcal{S}_{\Gamma}$, $U^{S}:=(u_i^S)_{i\in N}$, $\forall X\in N\backslash i,u_i^S(X)\equiv u_i(X\cap V_i)$, $ u_i^S(X\cup i):=u_i((X\cap V_i) \cup i)$. $\Gamma^S:=\Gamma(N,A,U^S)$ is still a monotone game. If further $S\in \mathcal{MS}_{\Gamma}$, then $\Gamma^S$ is a weakest link game. The following lemma is a direct application of the monotone comparative statics above.

\begin{lemma}\label{lem:3}
Note that tree-depth $td(G(N,E))$ is non-decreasing in $G$, so $\forall S\in \mathbb{S}_{\tau}$, $\exists S'\in \mathbb{MS}_{\tau}$ such that $S'$ is the subgraph of $S$, and the monotone comparative statics implies $\forall X\in 2^N, \tau_{\Gamma}(X) \leq \tau_{\Gamma^S}(X) \leq \tau_{\Gamma^{S'}}(X)=td(S'|_{W(X)})\leq td(S|_{W(X)})$. 
\end{lemma}

The following result builds the equivalence between general synchronous games and weakest link games\footnote{Note that people should not expect strategic equivalence in general, as the weakest link game adds additional structure to the feasible NE. For instance, there exists no weakest link game whose pure strategy NE set is: $\{\emptyset, \{1,2,3\},\{1,4,5\}, N\}$.}.

\begin{proposition}\label{thm:2}
For any synchronous game with stage game $\Gamma(N,A,U)$, there exists at least one weakest link game $\Lambda(N,E,A,P)$ such that: $\forall X\in 2^N, \tau_{\Gamma(N,A,U)}(X)\equiv \tau_{\Lambda(N,E,A,P)}(X)=td(G(N,E)|_{W(X)})$. The graph $G$ is a (minimal) sufficient network of $\Gamma$.
\end{proposition}

The following result is a direct corollary of \autoref{thm:2} and \autoref{lem:3}, which gives a second algorithm to solve $\tau_{\Gamma}$, a two steps procedure: find all of the sufficient networks and compute the tree-depth of each of them. This approach is more convenient when the sufficient network class is relatively small or the tree-depth is easy to compute.
\begin{corollary}\label{prop:5}
 Alternative algorithm: $\forall X\in 2^N,\tau_{\Gamma}(X) \equiv \text{min}_{S\in \mathbb{MS}_{\Gamma}}td(S|_{W(X)})$.
\end{corollary}

The equivalence between the synchronous game and the asynchronous design is obvious now: since both of the problems can be turned into the search of the greatest sufficient graph with the least tree-depth. And as we all know, it is more convenient to compute the operator $\tau_{\Gamma}$ recursively. We conclude it as the following theorem.

\begin{theorem}
For a stage game $\Gamma(N,A,U)$ and the number of stages $T$ fixed, and the least MSPNE outcome of the optimal asynchronous coordination design coincides with the least MSPNE of the monotone game, and the partition that maximizes the least MSPNE outcome of the asynchronous game coincides with the policy function of the corresponding synchronous game.
\end{theorem}

\section{Extension}\label{sec:5}
In this section, we will discuss several important questions in practice. 
First, we apply the asynchronous design to explain the redemption restriction of the mutual fund. Second, though the precise computation of $\tau_{\Gamma}$ is NP-hard in general, there is still a substantial class of games that can be computed efficiently, which also has implications on the asynchronous design. Lastly, from the policy intervention view, we will show that a principal's capacity of facilitating the adversarial coordination performance via subsiding the key player is relatively limited. 

\subsection{Costly Optimal Asynchronous Design}
 We state the robust asynchronous coordination design problem faced by a principal formally.

Players: A principal $I$ and a finite set $N$ of players. Actions: $A_I=\bigcup \{(T,\{ N_t\}_{t=1}^{T}\}$, in which $T\in \mathbb{N}^+, \{N_t\}_{t=1}^{T}$ is a partition of $N$. $\forall i\in N, A_i=\{0, 1\}$. Timeline: At the beginning of the game, the principal chooses an asynchronous coordination structure $a_I \in A_I$, including the number of periods and a corresponding partition. The principal's choice is common knowledge. Then the players take the asynchronous coordination game. At the end of the game, the outcome is realized. Preferences: The depositors' utilities satisfy the assumptions for the stage game above. Let $\mathcal{P}$ be the set of preferences that are weakly increasing in $\text{min}_{a^*\in \text{MSPNE Outcome}} a^*$, strictly decreasing in $T$, and independent of other variables. The principal's preference just needs to satisfy that: $P_I\in \mathcal{P}$. 

In words, the principal chooses an asynchronous structure to maximize the least MSPNE outcome corresponding to the asynchronous structure while minimizing the number of partitions. 

In practice, the game above can be seen as an optimal debt maturity structure design: the bank has several atomic firms as depositors, which are interrelated in financial risk sharing through cross-ownership, for example. Now the bank wants to eliminate the risk of bank runs by diversifying the feasible withdrawal date: the depositor can only schedule a withdrawal on a specific date during a time zone periodically\footnote{Such a design is more common in open-end funds, which has unambiguous restriction on redemption window}. This also explains why the principal has a strictly decreasing preference for the length of the game since the principal has to compensate the players for inflexibility. Also see \cite{basak2022panics} for a reference.

Surprisingly, we will find that although the principal's preference space is abundant, the optimal robust asynchronous coordination design is relatively narrow.

\begin{proposition}
An asynchronous structure $\{(T,\{ N_t\}_{t=1}^{T}\}$ is optimal for some preferences in $\mathcal{P}$ if and only if $\phi_{\Gamma}(T)\supsetneq \phi_{\Gamma}(T-1)$, and $a^*(\Gamma(N,A,U,T,\{ N_t\}_{t=1}^{T}))=\phi_{\Gamma}(T)$. 
\end{proposition}

\begin{proof}
The proof is self-contained. If $\phi_{\Gamma}(T)= \phi_{\Gamma}(T-1)$, then there exists a partition $\{ N_t\}_{t=1}^{T-1})$ such that $a^*=\phi_{\Gamma}(T-1)$, which is better for the principal. Conversely, one can easily create some extreme preferences like lexicographic order\footnote{Since the subjects to be compared are finite, it always allows for a utility representation.}. 
\end{proof}

\begin{corollary}
The number of different $T$ that can be optimal for some $P_I\in \mathcal{P}$ equals to $1+$the number of $t+1, t\in \mathbb{N}^+$ such that $\phi_{\Gamma}(t+1)\supsetneq \phi_{\Gamma}(t)$, which is at most at most $1+\lfloor \sqrt{2N+\frac{9}{4}}-\frac{3}{2} \rfloor$.
\end{corollary}

To see this, note that if $\phi_{\Gamma}(a)\supsetneq \phi_{\Gamma}(b)$, then $|\phi_{\Gamma}(a)-\phi_{\Gamma}(b)|\geq a$, and calculate $2+3+...+T+1 \leq |N|$, which implies the desired result. We can construct such games easily via the weakest link game on a graph consisting of several disconnected cliques with $2,3,4,...$ players, so the bound is tight. 

Therefore, regardless of the preference of the policymaker, to determine the optimal asynchronous design, we only need to compare these non-trivial $T$ candidates.

\subsection{Centrality Measure} 
Whether one player always takes higher action than the other player in the least MSPNE of the synchronous game and asynchronous design, not surprisingly, is characterized by the operator $\tau_{\Gamma}$. The following result can be seen as a dynamic game version of \cite{sadler2022ordinal}. These results are almost obvious given the operator $\tau_{\Gamma}$ is well-defined.

\begin{definition}
Weak centrality $w$: $\forall i,j\in N$, we say $i \succeq^{w} j$ if $\tau_{\Gamma}(i)\leq \tau_{\Gamma}(j)$. Strong centrality $s$: $\forall i,j\in N$, we say $i \succeq^{s} j$ if in all NE of the stage game $\Gamma$, $a_i\geq a_j$.
\end{definition}

According to the definition, weak centrality is a weak total order, and strong centrality is a weak partial order. The following proposition shows that they coincide with the ``centrality'' defined by \cite{sadler2022ordinal}. The following proposition shows that the weak centrality of the monotone game can be determined through the operator $\tau_{\Gamma}$, while the strong centrality of the monotone game agrees with the strong centrality of the stage game. 

\begin{proposition}
\begin{enumerate}
    \item Strong centrality implies weak centrality. 
    \item $i \succeq^{w} j$ iff $\forall T\in \mathbb{N}^+$, in the least MSPNE outcome of the synchronous game $\Gamma(T)$, $a_i\geq a_j$, iff in the least MSPNE outcome of the $T$-optimal asynchronous design, $a_i\geq a_j$.
    \item $i \succeq^{s} j$ iff $\forall T\in \mathbb{N}^+$, in any of the MSPNE outcome of the synchronous game $\Gamma(T)$, $a_i\geq a_j$, iff in any of the MSPNE outcome of the $T$-optimal asynchronous design, $a_i\geq a_j$.
\end{enumerate}
\end{proposition}

Note that the weak centrality is a weak total order, while it cannot involve too many strictly different weak centrality classes: it coincides with the number of $T\in \mathbb{N}^+$ such that $\phi_{\Gamma}(T+1)\supsetneq \phi_{\Gamma}(T)$, which is at most at most $\lfloor \sqrt{2N+\frac{9}{4}}-\frac{3}{2} \rfloor$.

Additionally, we should also notice that the centrality measure is not too sharp\footnote{Though some network researchers may make a contrasting argument, see \cite{bloch2023centrality}.}, especially for the analysis of coordination games. For instance, given a weakest link game on a strongly connected network, every vertex has the same strong centrality. This motivates the invention of stronger centrality alternates. For instance, the policy function of $\tau_{\Gamma}$ can also be used as a centrality measure. The basic intuition is that tree-depth has been used to characterize the difference between a graph and the star graph, while $\tau_{\Gamma}$ can be seen as the expansion of the tree-depth.

\subsection{Ordered Game}

Given the NP-hardness of $\tau_{\Gamma}$ in general, a natural question is, what kinds of games can be computed efficiently? The first category is the weakest link game with specific graphs, given tree-depth has a closed-form solution. The second category we find is ordered game, where the players' incentives to take action $1$ and the influences to incentivize others to take action $1$ can be totally ordered. This approach is similar to \cite{sadler2024games}.

\begin{definition}
A stage game $\Gamma$ is cost-ordered if: $\forall i,j\in N, i<j, X\subset N\backslash\{i,j\}$, $u_j(j\cup X)> u_j(X)$ implies that there exists a finite sequence $\{n_l\}_{l=1}^{k}, n_1=j,n_{k}=i$, $\forall 1\leq p\leq k, u_{n_p}(\bigcup_{1\leq t\leq p} n_p \cup X )>u_{n_p}(\bigcup_{1\leq t\leq p-1} n_p \cup X )$. A stage game $\Gamma$ is strongly cost-ordered if: $\forall i,j\in N, i<j, X\subset N\backslash\{i,j\}$, $u_j(j\cup X)> u_j(X)$ implies $u_i(i\cup X)> u_i(X)$. A stage game $\Gamma$ is contribution-ordered if: $\forall i,j,k\in N,i<j, X\subset N\backslash \{i,j,k\}$, 
$u_k(i\cup X)> u_k(X)$ implies $u_k(j\cup X)> u_k(X)$. A stage game $\Gamma$ is contribution-natural if: $\forall i,j,k\in N,i\neq j, X\subset N\backslash \{i,j,k\}$, 

$u_k(i\cup X)> u_k(X)$ implies $u_k(j\cup X)> u_k(X)$.
\end{definition}

If a stage game $\Gamma$ is cost-ordered, then $\forall X\in SSE, X=\{1,2,...,k_{X}\}$. So the NE set is totally ordered in set cover order. If a stage game $\Gamma$ is contribution-natural, then whether $u_i(X\cup i)>u_i(X)$ depends on $|X|$ only. Therefore, such a stage game always can be strongly cost-ordered after permutation. In this sense, any stage game that is contribution-natural is $\tau-$equivalent to an aggregative game with heterogeneous cost.
 
\begin{proposition}\label{prop:8}
If a stage game is both cost-ordered/strongly cost-ordered and contribution-ordered, then suppose the utility value is saved, then $\tau(\Gamma)$ can be computed efficiently in quadratic/linear time: $\tau_{\Gamma}(X)=\text{min}_{Y\in SSE(\Gamma): Y\supset X}\tau(\underline{\Gamma}(Y))$,

\begin{footnotesize}

   \begin{equation*}	   \tau(\Gamma)=
 		\begin{cases}
 \tau(\overline{\Gamma}{(N\backslash \{1,2,...,k\})}) & \text{if } u_k(k)>u_k(\emptyset)/u_1(1)>u_1(\emptyset)\\
    1 & \text{else if } N=\emptyset\\
1+  \tau(\overline{\Gamma}(N\backslash |N|)) & \text{else } 
 		\end{cases}
 	\end{equation*} 
\end{footnotesize}

If a stage game is both strictly cost-ordered and contribution natural, $\tau(\Gamma)$ can be further computed efficiently in linear time $(O(|N|))$ after acceleration, see \autoref{alg:1}. $c_i$ is the minimum number of other players taking action 1 such that $i$ prefer taking action $1$ strictly: $\forall i\in N, \forall X\subset N\backslash i, u_i(X\cup i)>u_i(X)$ if and only if $|X|\geq c_i, c_i\in \mathbb{N}^+, 1\leq c_i\leq c_{i+1}\leq n-1$.
 \end{proposition}

\begin{algorithm}[h]
\renewcommand{\algorithmicrequire}{\textbf{Input:}}
	\renewcommand{\algorithmicensure}{\textbf{Output:}}
\caption{Accelerated Algorithm}\label{alg:1}
\begin{algorithmic}[1]
\Require thresholds $(c_i)_{i=1}^{n}, c_i\in \mathbb{N}^+, 1\leq c_i\leq c_{i+1}\leq n-1$

\State $t \gets 0$
\State $l \gets 1$
\State $d \gets 0$
\State $r \gets n$
\While{$l<r$}
\If{$c_l \leq d$}
\State $ l \gets l+1, d \gets d+1$
\Else
\State $r \gets r-1, t \gets t+1, d \gets d+1$
\EndIf
\EndWhile
\Ensure $\tau(\Gamma)=t+1$
\end{algorithmic}
\end{algorithm}

Note that if the stage game is an ordered game, then the corresponding optimal asynchronous design also shares a much simpler form, due to the structure of the policy function. Now we derive three examples of ordered games.

\begin{example}[Aggregative Game with Heterogeneity]
A stage game is an aggregative game if $\forall i\in N, u_i=f_i(a_i,\sum_{j\in N}a_{j}), f_i:\{ 0,1 \}\times \mathbb{N}^+\rightarrow \mathbb{R}$. Then there exists unique $c_i\in \mathbb{N}^+, 1\leq c_i\leq |N|-1$ such that $f_i(0,c_i)> f_i(1,c_i), f_i(0,c_i+1)< f_i(1,c_i+1)$. We can further assume that $\forall i<j\in N, c_i\leq c_j$. An aggregative game is strongly cost-ordered and contribution natural. 

An action profile a is a NE of the stage game iff $X_a=\{1,2,...,k\}$, $c_{k}\geq k-1, c_{k+1}>k$. When the game is homogeneous, $\forall i\in N, c_i\equiv k$, it reduces to the classic result that $\tau_{\Gamma}(i)\equiv \tau_{\Gamma}=k+1$. Given $\tau({\Gamma})\leq k+1$, the greatest threshold vector $c$ can be $\overline{c}=(k,k+1,k+2,...,|N|-2,|N|-1,|N|-1,...,|N|-1)$.
\end{example}

\begin{example}[Weakest Link Game with Aligned Nested Split Graph]
A directed graph $G=(N,E)$ is an aligned nested split graph if $\forall i,j\in N, i<j, N_{in}(i)\subset N_{in}(j)\cup j, N_{out}(i)\subset N_{out}(j)\cup j$,  in which $N_{in}(i):=\{j\in N: (j,i)\in E\}, N_{out}(i):=\{j\in N: (i,j)\in E \}$. The weakest link game on an aligned nested split graph is cost-ordered and contribution-ordered. 

A directed graph is an aligned nested split graph iff $N_{in}(i)=\{I_i,I_i+1,I_i+2,...,|N|-1,|N|\}\backslash i, N_{out}(i)=\{O_i,O_i+1,O_i+2,...,|N|-1,|N|\}\backslash i$, in which $\forall i<j\in N, I_i\geq I_j \text{ or } I_i=j, I_j=j+1$; $ O_i\geq O_j \text{ or } O_i=j, O_j=j+1$.
\end{example}

\begin{example}[Thresholds Game with Opposed Nested Split Graph]
Thresholds game on a directed graph $G$: $u_i(a_i=1,a_{-i})>u_i(a_i=0,a_{-i})$ iff $\sum_{j\in N_{in}(i)} a_j\geq k_i$, in which $k_i\in \mathbb{N}^+, |N_{in}(i)|\geq k_i$. A directed graph $G$ is an opposed nested split graph if $\forall i,j\in N, i<j, k_i\leq k_j, N_{in}(j)\subset N_{in}(i)\cup i, N_{out}(i)\subset N_{out}(j)\cup j$. Thresholds game on an opposed nested split graph is strongly cost-ordered and contribution-ordered. 

A directed graph is an opposed nested split graph iff $N_{in}(i)=\{I_i,I_i+1,I_i+2,...,|N|-1,|N|\}\backslash i, N_{out}(i)=\{1,2,3,...,O_i-1,O_i\}\backslash i$, in which $\forall i<j\in N, I_i\leq I_j \text{ or } I_i=i+1, I_j=i$; $ O_i\leq O_j \text{ or } O_i=j, O_j=j-1$.
\end{example}

\begin{figure}[h]
  \centering
  \begin{tikzpicture}[>=stealth, node distance=2.5cm, every node/.style={circle, draw, minimum size=20pt}]
    
    \node (1) at (90:3) {1};
    \node (2) at (30:3) {2};
    \node (3) at (330:3) {3};
    \node (4) at (270:3) {4};
    \node (5) at (210:3) {5};
    \node (6) at (150:3) {6};

    \draw[->] (6) -- (1);
    \draw[->] (6) -- (2);
    \draw[->] (6) -- (3);
    \draw[->] (6) -- (4);
    \draw[->] (6) -- (5);
    
    \draw[->] (5) -- (1);
    \draw[->] (5) -- (2);
    \draw[->] (5) -- (3);
    \draw[->] (5) -- (4);
    \draw[->] (5) -- (6);
    
    \draw[->] (4) -- (1);
    \draw[->] (4) -- (2);
    
    \draw[->] (3) -- (1);
    \draw[->] (3) -- (2);
  \end{tikzpicture}
  \caption{Weakest Link Game with Aligned Nested Split Graph}
  \label{fig:4}
\end{figure}
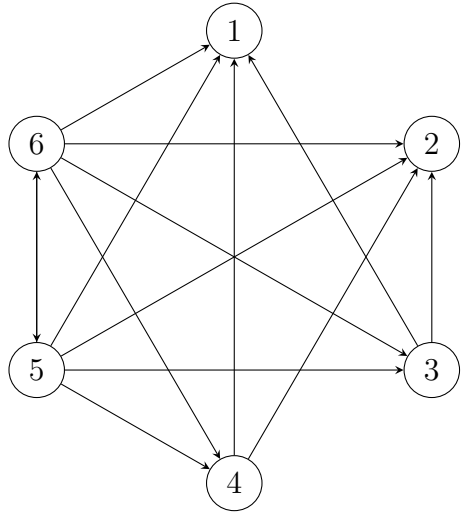

\subsection{Policy Intervention}
 
A natural conjecture is that, given the cruciality of some key players, like the central player in a weakest link game with a star network, the policymaker may facilitate the adversarial coordination performance significantly via simple policy intervention, like subsiding a few key players. The answer to the question includes typically two parts. 

From a pessimistic view, if the policymaker wants to eliminate all inefficient equilibrium outcomes only, the policymaker cannot do much too better than the adversarial coordination performance via policy intervention. The incentive is that given the game is of common interests and strategic complementarities, all of the players prefer the greatest equilibrium. If a player is crucial, then the player himself must understand this in equilibrium. Therefore, there are few spaces to improve it via outside intervention. More formally, subsiding 1 player can reduce the number of stages that an adversarial approach needs at most $1$: $\forall i\in N, \tau_{\overline{\Gamma}(N\backslash i)}(N\backslash i) \leq \tau_{\Gamma}(N)\leq \tau_{\overline{\Gamma}(N\backslash i)}(N\backslash i)+1$. 

From an optimistic view, however, the policymaker may facilitate the local performance significantly. That is to say, the inequality on the right-hand-side holds for $N$, but may not hold for some $X\subsetneq N\backslash i$. For instance, in the following weakest link game, though $\forall i\in N,\tau_{\Gamma}(\{5,6,7,8,9 \})=4$, $\tau_{\overline{\Gamma}(N\backslash 1)}(\{5,6,7,8,9 \})=1$. The reason is that: after subsiding some player $i$, there may emerge some new strictly sufficient sets. Originally,  player $i$ may take action $1$ under restrictive conditions only. In the example, for instance, the corresponding graph of the weakest link game restricted on the sub-group players $\{ 1,2,3,4\}$ is a complete graph, which takes a long time to coordinate efficiently. While the corresponding graph of the weakest link game restricted to the sub-group players $\{ 1,5,6,7,8,9\}$ is a star network, which is relatively fast then. 

Therefore, it is crucial to evaluate the effect of subsiding a set of players $X$ case by case: $\phi_{\overline{\Gamma}(N\backslash X)}(T) \backslash \phi_{\Gamma}(T)$. We do not have any bounds in general.

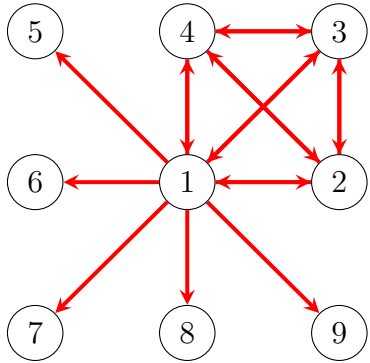
\begin{figure}[h]
  \centering
  \begin{tikzpicture}[>=stealth, node distance=2.5cm, every node/.style={circle, draw, minimum size=20pt}]
    
 \node    (1) at (0,0) {1};
\node (2) at (2,0) {2};
\node (3) at (2,2) {3};
\node (4) at (0,2) {4};
\node (5) at (-2,2) {5};
\node (6) at (-2,0) {6};
\node (7) at (-2,-2) {7};
\node (8) at (0,-2) {8};
\node    (9) at (2,-2) {9};

    \draw[->, red,line width=1.5pt] (1) -- (2);
    \draw[->, red, line width=1.5pt] (1) -- (3);
    \draw[->, red, line width=1.5pt] (1) -- (4);
    \draw[->, red, line width=1.5pt] (1) -- (5);
    \draw[->, red, line width=1.5pt] (1) -- (6);
    \draw[->, red, line width=1.5pt] (1) -- (7);
    \draw[->, red, line width=1.5pt] (1) -- (8);
    \draw[->, red,line width=1.5pt] (1) -- (9);
    
    \draw[->, red, line width=1.5pt] (2) -- (1);
    \draw[->, red, line width=1.5pt] (2) -- (3);
        \draw[->, red, line width=1.5pt] (2) -- (4);
    \draw[->, red, line width=1.5pt] (3) -- (1);
    \draw[->, red, line width=1.5pt] (3) -- (2);
        \draw[->, red, line width=1.5pt] (3) -- (4);
    \draw[->, red, line width=1.5pt] (4) -- (1);
    \draw[->, red, line width=1.5pt] (4) -- (2);
    \draw[->, red, line width=1.5pt] (4) -- (3);
  \end{tikzpicture}
  \caption{Policy Intervention of the Weakest Link Game}
  \label{fig:8}
\end{figure}

\subsection{Robustness}

We briefly discuss the robustness of the synchronous game. 

First, in the benchmark model, we assume the payoff depends only on the action profile in the last period to make the model elegant. But we could allow a less restrictive, time-dependent payoff at little cost. In the following, the utilities of the players are $U=(u_i)_{i\in N}: \prod_{1\leq t\leq T}\prod_{i\in N}A_{it} \rightarrow \mathbb{R}^{|N|}$.

\begin{proposition}
Under the following two conditions, all of the results hold:
\begin{enumerate}
\item Nearness: $\exists U^* =(u_i^*)_{i\in N}: \prod_{i\in N}A_i \rightarrow \mathbb{R}^{|N|} $ such that $\forall i\in N, \forall a,b\in \prod_{i\in N}A_i$, $u_i^*(a)>(=)u_i^*(b) \Rightarrow $ $\forall \{a_t\}_{t=1}^{T}, \{b_t\}_{t=1}^{T}\in \prod_{t=1}^{T}\prod_{i\in N}A_{it}$ such that $a_T=a,b_T=b$, $u_i(\{a_t\}_{t=1}^{T})>(=)u_i(\{b_t\}_{t=1}^{T})$.
\item Procrastination: $\forall i\in N, \forall \{a_{-it}\}_{t=1}^{T}\in \prod_{t=1}^{T}\prod_{j\in N\backslash i}A_{jt}$, \\$\text{argmax}_{\{a_{it}\}_{t=1}^{T}\in \{A_{it}\}_{t=1}^{T}}$$u_i(\{a_{it}\}_{t=1}^{T},\{a_{-it}\}_{t=1}^{T}) \neq 1^{T}$.
\end{enumerate}
\end{proposition}

Nearness is a natural assumption, it implies that the payoff does not depend on the previous history too much. Procrastination further makes sure that \autoref{lem:1} holds, without which we fail to show the equilibria convergence. Given both assumptions hold, all of the proofs above maintain, so we get the desired results.

For a non-trivial example, the lag-complementary model satisfies our assumptions: $\forall i\in N, u_i(a_i=0, a_{-i})\equiv 0$, $u_i(a_i=1,a_{-i})$ is non-decreasing in $a_{-i}$, and $\forall a\neq b, u_i(a)\neq u_i(b)$. $\forall \{ a_t\}_{t=1}^{T} \in \{ A_t\}_{t=1}^{T}, a_{-1}=a_{0}:=0^{N}, U_i(\{a_t\}_{t=1}^{T})=\sum_{t=1}^{T}\delta^{t-1} [u_i(a_{it}, a_{-i,t-1})-u_i(a_{i,t-1}, a_{-i,t-2})]$, $0< \delta<1, \delta \rightarrow 1$.

Moreover, people may argue that the synchronous game is not symmetric: only the action $1$ can be committed. However, even if players could pledge the socially inefficient action $0$, it is weakly dominated given all of the players adopt the monotonic pure strategy only\footnote{So if there is further any cost of committing 0, it is strictly dominated}. More formally, the monotone game $\Gamma'(N,A',U,T)$, in which the irreversible constraint turns into the following. The new order of actions is given by $1\succ w\succ 0$, and equipped with the product order if needed. The monotonic strategy implies that each player's strategy maps a higher history into a higher action. Once a player commits to action $1$ or action $0$, he cannot withdraw the commitment; while the player always has the option to wait($w$) without making any commitment during the pledging stage.

$$\forall 1\leq t\leq T,   A_{it}'=
 		\begin{cases}
\{w,0,1\} & \text{if } t=1 \vee [(a_{it-1}=w) \wedge (t\neq T)] \\
\{ 0\} & \text{else if } a_{it-1}=0\\
    \{1\} & \text{else if } a_{it-1}=1\\
\{0,1 \} & \text{else } \\
 		\end{cases}$$

To see why pledging the socially inefficient action $0$ is weakly dominated during the commitment stage if all of the players adopt the monotonic pure strategy only, the player could always choose to wait instead of committing $0$ among all commitment stages and adopt the most conservative strategy at the terminal. $w$ is a higher action compared with $0$, therefore, other players' actions are at least as high as before, which yields a payoff at least as high as before, which proves the desired result\footnote{The argument holds only under the assumption that all of the players adopt the monotonic pure strategy. Otherwise, one can easily construct some weird strategies that collectively map a lower history into a higher action, and the argument above will fail.}.

\section{Conclusion}\label{sec:colc}

This paper provides a comprehensive analysis of the equivalence between synchronous and asynchronous coordination mechanisms in dynamic games characterized by strategic complementarities and common interests. Despite the structural differences between these coordination modes, we have shown that their outcomes converge under pessimistic evaluation conditions, specifically when focusing on the least favorable equilibrium outcome (MSPNE). The equivalence we establish demonstrates that the greatest implementable outcome in the asynchronous game matches that of the synchronous setting, suggesting that the timing structure does not affect the achievable welfare as long as the critical conditions are met.

We offer novel characterizations for both coordination settings: a recursive approach for synchronous coordination and a graphical representation for asynchronous coordination. These results contribute to a deeper understanding of how different timing structures affect coordination outcomes. The ability to bridge these two approaches enriches our understanding of coordination dynamics, particularly where commitment and strategic timing significantly influence behavior.

Additionally, we address the computational aspects of achieving these outcomes, revealing that the complexity of the coordination problem is closely tied to Tree-depth, rendering it NP-Hard in general. Nevertheless, we identify a significant class of ordered games for which an efficient dynamic programming solution exists, making practical implementation feasible in certain scenarios. This tractability is particularly useful when designing optimal coordination strategies that depend on precise timing structures.

Our findings also inform the optimal design of sequential coordination mechanisms. By examining the principal's role in setting preferences and constraints, we demonstrate that the design of the sequence and the number of partitions can be approached through a simplified cost-benefit analysis. These insights are particularly valuable in environments where the strategic timing of decisions can significantly enhance efficiency and effectiveness.

In sum, our research advances the understanding of dynamic coordination through a unified framework that addresses both synchronous and asynchronous settings. By establishing equivalence between these coordination mechanisms, we provide new perspectives on how efficient outcomes can be attained across various economic contexts without reliance on specific timing structures. This work offers practical guidance for the design of mechanisms that enable effective coordination under strategic complementarities.

\newpage
\bibliography{ref}

\bibliographystyle{aer}

\newpage
\appendix

\section{Omitted Examples}\label{sec:counterexample}

\subsection{Assumptions: Discussion}
In this section, we show that each of our assumptions cannot be removed. 
\paragraph{Common interests cannot be removed:}

\begin{example}
Consider the stage game with 3 players. $N=\{ 1,2,3\}. u_1(a)=(1-a_1)[2(1-a_2)(1-a_3)-1], u_2(a)=a_2(2a_1-1), u_3(a)=a_3(2a_1-1)$. In this example, player 1 does not prefer the greatest equilibrium. 

In the synchronous game with $T=2$, $O(\Gamma(2))=\emptyset$. In the asynchronous design with $T=2$, there exists no partition such that all players take action $1$ in all MSPNE.
\end{example}

\begin{example}\label{exm:ci}
Consider the stage game with 8 players. $N=\{1,2,3,4,5,6,7,8\}$. $u_1(a)= \max\{ a_1(4a_2-3+a_3a_4a_5+a_6a_7a_8), 2a_3a_4a_5 \}$, $u_2(a)= \max\{ a_2(4a_1-3+a_3a_4a_5+a_6a_7a_8), 2a_3a_4a_5 \}$, $u_3(a)= a_3(2a_4a_5-1), u_4(a)=a_4(2a_3a_5-1), u_5(a)=a_5(2a_3a_4-1)$, $u_6(a)=a_6(2a_7a_8-1), u_7(a)=a_7(2a_6a_8-1), u_8(a)=a_8(2a_6a_7-1)$. In this example, the tie-break rule is violated.

In the synchronous game with $T=2$, the MSPNE outcome is not a lattice: $O(\Gamma(2))=\{\{1,2\}, \{3,4,5\}, \{1,2,3,4,5\}, \{1,2,6,7,8\}, N \}$. In the asynchronous game with $T=2$, there exists no partition such that players 1 and 2 take action $1$ in all MSPNE.
\end{example}

\paragraph{Deviation-proof cannot be removed:}
\begin{example}\label{exm:dp}
Consider the stage game with two players. The payoff matrix is given in \autoref{tab:1}.\footnote{It has both some features like prisoner's dilemma and matching pennies: player $1$ has a strictly dominant strategy $0$, and player $2$'s best response is to take the same action as player $1$, which explains why we call such a game ``mixed game''.} In this example, if player $1$ could commit to taking action $1$ creditably, then his payoff would be higher than the Nash equilibrium payoff, and the deviation-proof condition is therefore violated.

In the synchronous game with $T=2$, the unique MSPNE outcome is that both players take action 1. In the asynchronous game with $T=2$, under the partition $N_1=\{ 1\}, N_2=\{2\}$, the unique MSPNE outcome is that both players take action 1.
\end{example} 

\begin{table}[h]
\centering
 \begin{tabular}{|c|c|c|}
    \hline
$1 \backslash 2$ & 0 & 1 \\ \hline
0 & 1, 1 & 3, 0 \\ \hline
  1 & 0, 2 & 2, 3 \\ \hline 
    \end{tabular}
\caption{Payoff Matrix of \autoref{exm:dp}.}
\label{tab:1}
\end{table}

Moreover, when more than two players's utilities violate the deviation-proof condition, the analysis is in general intractable\footnote{One may ask whether we could replace it with an additional non-degenerate condition which claims that nobody has a dominant strategy. Unfortunately, to study the local behavior of the monotone game precisely, which is the main target of this paper, this assumption is inevitable.}.

\begin{example}\label{exm:dp3}
Consider the stage game with 3 players. $N=\{1,2,3 \}$. $ u_1(a)=2a_3-a_1, u_2=2a_3-a_2$, $u_3(a)=a_3[2\max\{a_1, a_2\}-1]$. Easy to see players 1 and 2 each have a strictly dominant strategy $0$, and player $3$'s best response is to take the greater action that player $1$ and player $2$ take. Though the stage game satisfies the single-crossing condition, the exists a free-rider incentive like the public good game: if player $1$ does not pledge action $1$, then player $2$ can benefit from committing action $1$; if player $1$ pledged action $1$, then player $2$ has no incentive to commit action $1$. The same thing holds for player $1$. So in some sense, the game between players 1 and 2 is a game of strategic substitutes.

In the synchronous game with $T=2$, there exists no least MSPNE outcome, as the two minimal MSPNE outcomes are player 1 and player 3 take action 1, and player 2 and player 3 take action 1. In the asynchronous game with $T=2$, for partition $N_1=\{1,2\},N_2=\{3\}$, there exists no least MSPNE outcome. Also, among all partitions that guarantee the existence of the least MSPNE outcome, there exists no partition that realizes the greatest least MSPNE outcome. Since there exists partition $N_1=\{1\}, N_2=\{2,3\}$ such that player 1 and player 3 take action 1 in all MSPNE, and partition $N_1=\{2\}, N_2=\{1,3\}$ such that player 2 and player 3 take action 1 in all MSPNE, but there exists no partition such that all players take 1 in all MSPNE.

\end{example}

\subsection{Solution Concept: Discussion}

In this section, we will compare MSPNE with SPNE. We show that SPNE could be less appealing for both synchronous and asynchronous games, due to several reasons.

\paragraph{A SPNE outcome of the asynchronous game is not necessarily a NE of the stage game:}

\begin{example}
Consider the stage game with 3 players: $N=\{1,2,3 \}$, $u_1(a)=a_1+a_2+a_3, u_2(a)=a_2(2a_3-1), u_3(a)=a_3(2a_2-1)$. Player 1 has a strictly dominant strategy $1$. In the asynchronous game $(T=2, P_1=\{1\}, P_2=\{2,3\})$, the following strategy is a SPNE but the SPNE outcome is not a NE of the stage game: $a_1=0$, $s_2(a_1=1)=s_3(a_1=1)=0, s_2(a_1=0)=s_3(a_1=0)=1$. 
\end{example}

\paragraph{The nonexistence of least SPNE outcome of the synchronous and asynchronous game}

\begin{example}
Consider the stage game with 5 players. $N=\{1,2,3,4,5\}$. $u_1(a)=a_1(2a_2-1)+2a_3a_4a_5, u_2(a)=a_2(2a_1-1)+2a_3a_4a_5, u_3(a)=a_3(2a_4a_5-1), u_4(a)=a_4(2a_3a_5-1), u_5(a)=a_5(2a_3a_4-1)$. 

In the synchronous game with $T=2$, there exists no least SPNE outcome: both players 1,2 take action 1 and others take action 0, and players 3,4,5 take action 1 and others take action 0 are SPNE outcomes, while no player take action 1 is not a SPNE. In the asynchronous game with $T=2, N_1=\{1\}, N_2=\{2,3,4,5\}$, there exists no least SPNE outcome: both players 1,2 take action 1 and others take action 0, and players 3,4,5 take action 1 and others take action 0 are SPNE outcomes.
\end{example}

\paragraph{The nonexistence of partition that maximizes the least SPNE outcome of the asynchronous game}

\begin{example}\label{exm:asyngame}
Consider the stage game with 7 players. $N=\{1,2,3,4,5,6,7 \}$. $u_1(a)=a_1(2a_2a_3-1)+2a_5, u_2(a)=a_2(2a_1-1), u_3(a)=a_3(2a_1-1)$, $u_4(a)=a_4(2a_1a_5a_6a_7-1), u_5(a)=a_5(2\max\{a_4,a_6a_7 \}-1), u_6(a)=a_6(2\max\{a_4,a_5a_7 \}-1), u_7(a)=a_7(2\max\{a_4,a_5a_6 \}-1)$.

It is easy to see that $P_1=\{1,4\}, P_2=\{2,3,5,6,7 \}$ is the unique partition such that all players take action 1 in all MSPNE of the corresponding asynchronous game when $T=2$. Nevertheless, player $5,6,7$ take action 1 and all others take action 0 is also a SPNE outcome according to the following strategy $s$: $a_{1}=a_{4}=0$, $s_2(a_1=1)=s_3(a_1=1)=1, s_2(a_1=0)=s_3(a_1=0)=0$, $s_5(a_4=1)=s_6(a_4=1)=s_7(a_4=1)=1$, $s_5(a_1=0,a_4=0)=s_6(a_1=0,a_4=0)=s_7(a_1=0,a_4=0)=1$, $s_5(a_1=1,a_4=0)=s_6(a_1=1,a_4=0)=s_7(a_1=1,a_4=0)=0$. Additionally, the least SPNE of the asynchronous game corresponding to the partition $P_1=\{1,4,5,6,7\}, P_2=\{2,3\}$ is player $1,2,3$ take action 1 and all others take action 0. Therefore, there exists no partition that maximizes the least SPNE of the corresponding asynchronous game.
\end{example}

\paragraph{SPNE is not reducible:}

MSPNE is reducible in the synchronous game in the sense that if a set of players $X$ take action 1 in all MPSNE, then we can without loss of generality consider the game $\overline{\Gamma}_{N\backslash X}(T)$, and $O(\Gamma(T))=X+O(\overline{\Gamma}_{N\backslash X}(T))=\{Y\cup X| Y \in O(\overline{\Gamma}_{N\backslash X}(T))\}$, which is a direct corollary of \autoref{prop:2}\footnote{Note that NE is reducible for the game of strategic complementarities, since the greatest and the least NE is exactly the greatest and least action profile that survives iterated elimination of strictly dominated strategy.}. Similarly, MSPNE is reducible in the asynchronous game. In fact, MSPNE is reducible even in the asynchronous design: if there exists a partition with $(T,\{N_t\}_{t=1}^{T}$ such that a set of players $X$ take action 1 in all MPSNE, then we can fix the partition $\{N_t\}_{t=1}^{T}\cap X$, and consider the game $\overline{\Gamma}_{N\backslash X}(T)$. This argument can be shown using the graphical representation.

Nevertheless, the argument does not hold for SPNE.

\begin{example}\label{exa:1}
Consider the stage game with 4 players. $N=\{1,2,3,4\}$. $u_1(a)=a_1(2a_2-1), u_2(a)=a_2(2a_1-1), u_3(a)=a_3(2a_2a_4-1), u_4(a)=a_4(2a_1a_3-1)$. It can be seen as a weakest link game on a directed graph in \autoref{fig:1}. Not hard to see $O(\Gamma)=\{ \emptyset, \{1,2\},N\}$. 

In the synchronous game $T=2$, there are two SPNE outcomes: players 1,2 take action 1, others take action 0, and all players take action 1. 
\end{example}

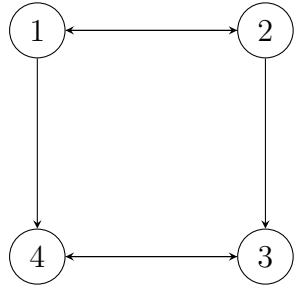
\begin{figure}[h]
  \centering
\begin{tikzpicture}[>=stealth, node distance=2cm, every node/.style={circle, draw, minimum size=20pt}]
    
    \node (1) at (0,0) {1};
    \node (2) at (3,0) {2};
    \node (3) at (3,-3) {3};
    \node (4) at (0,-3) {4};

    \draw[->] (1) -- (2);
    \draw[->] (1) -- (4);
    \draw[->] (2) -- (1);
    \draw[->] (2) -- (3);
    \draw[->] (3) -- (4);
    \draw[->] (4) -- (3);
  \end{tikzpicture}
 
  \label{fig:1}
\end{figure}

\begin{example}\label{exa:pq}
Consider the stage game with 4 players. $N=\{1,2,3,4\}$. $u_1(a)=a_1(2a_2-1), u_2(a)=a_2(2a_1-1), u_3(a)=a_3(2a_4-1)+2a_2, u_4(a)=a_4(2a_3-1)+2a_1$. $O(\Gamma)=\{ \emptyset, \{1,2\}, \{3,4\}, N\}$\footnote{In this example, the people may guess the two sub-games are ``independent'' so MSPNE should coincides with the SPNE. Unfortunately, it turns out to be wrong.}. 

In the synchronous game $T=2$, there are two SPNE outcomes: players 1,2 take action 1, others take action 0, and all players take action 1. 
\end{example}

To see the argument holds for asynchronous game, consider \autoref{exm:asyngame}. In the asynchronous game with $T=2, P_1=\{1,4\}, P_2=\{2,3,5,6,7 \}$, players 5,6,7 take action 1 in all SPNE.

\paragraph{SPNE is not robust to $\epsilon$-perturbation:}

\begin{example}

We construct two stage games for comparison. The first stage game involves 5 players, $N=\{ 1,2,3,4,5\}$. $\forall i=1,2,3, u_i(a)=2\min \{a_1,a_2,a_3 \}-a_i$, $u_4(a)=2\min \{a_1,a_4,a_5 \}-a_4$, $u_5(a)=2\min \{a_4,a_5 \}-a_5$. The stage game has 3 NE: $\{\emptyset, \{ 1,2,3\} , N \}$. 

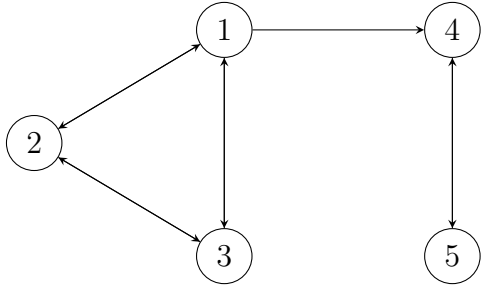
\begin{figure}[h]
  \centering
\begin{tikzpicture}[>=stealth, node distance=2cm, every node/.style={circle, draw, minimum size=20pt}]
    
    \node (1) at (0,0) {1};
    \node (2) at (-2.5,-1.5) {2};
    \node (3) at (0,-3) {3};
    \node (4) at (3,0) {4};
    \node (5) at (3,-3) {5};

    \draw[->] (1) -- (2);
     \draw[->] (1) -- (3);
     \draw[->] (1) -- (4);
    \draw[->] (4) -- (5);
    \draw[->] (5) -- (4);
    \draw[->] (2) -- (1);
    \draw[->] (2) -- (3);
    \draw[->] (3) -- (1);
    \draw[->] (3) -- (2);

  \end{tikzpicture}

  \caption{A Weakest Link Game without Spillovers}
  \label{fig:ex3}
\end{figure}

In the synchronous game with $T=3$, both $N$ and $\{1,2,3\}$ are SPNE outcomes. The outcome ${1,2,3}$ can be supported by the following SPNE: $\forall i\in N, s_{i1}=0$, $s_{i2}=1$ iff $a_{i1}=1$. $\forall i=1,2,3, s_{i3}=1$ iff $(a_{i2}=1) \vee (a_{12}+a_{22}+a_{32}\geq 2) \vee (a_{11}+a_{21}+a_{31}\geq 1) \vee (a_2=0^{N})$. $s_{43}=1$ iff $(a_{42}=1)\vee (a_{11}+a_{21}+a_{31}\geq 1) \vee ((a_{12}=1) \wedge (a_{52}=1)) \vee ((a_{22}+a_{32}=2) \wedge (a_{52}=1))$. $s_{53}=1$ iff $a_{42} \vee a_{52}=1$. 

The second stage game involves 5 players too, $N=\{ 1,2,3,4,5\}$. $u_1(a)=2\min \{a_1,a_2,a_3 \}-a_1+\epsilon a_4, \epsilon \in (0,1)$. $\forall i=2,3, u_i(a)=2\min \{a_1,a_2,a_3 \}-a_i$, $u_4(a)=2\min \{a_1,a_4,a_5 \}-a_4$, $u_5(a)=2\min \{a_4,a_5 \}-a_5$. The stage game has 3 NE: $\{\emptyset, \{ 1,2,3\} , N \}$. 

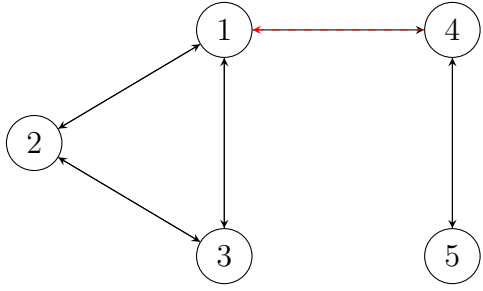
\begin{figure}[h]
  \centering
\begin{tikzpicture}[>=stealth, node distance=2cm, every node/.style={circle, draw, minimum size=20pt}]
    
    \node (1) at (0,0) {1};
    \node (2) at (-2.5,-1.5) {2};
    \node (3) at (0,-3) {3};
    \node (4) at (3,0) {4};
    \node (5) at (3,-3) {5};

    \draw[->] (1) -- (2);
     \draw[->] (1) -- (3);
     \draw[->] (1) -- (4);
    \draw[->] (4) -- (5);
    \draw[->] (5) -- (4);
    \draw[->] (2) -- (1);
    \draw[->] (2) -- (3);
    \draw[->] (3) -- (1);
    \draw[->] (3) -- (2);
\draw[dashed, ->, red] (4) -- (1);
  \end{tikzpicture}
  \caption{A Weakest Link Game with $\epsilon$-Spillovers}
  \label{fig:ex4}
\end{figure}

In the synchronous game with $T=3$, the unique SPNE outcome is $N$. Otherwise, the SPNE outcome candidate can be $\{1,2,3\}$ only. Then it must be the case that $a_{11}=0$, otherwise, all other players take action $1$ in equilibrium. While player $1$ has a strict incentive to deviate to $a_{i1}=1$, which is profitable for him since he receives the positive spillovers $\epsilon$ from player $4$'s action $1$.
    
\end{example}

\begin{example}
Consider the stage game with 7 players. $N=\{1,2,3,4,5,6,7 \}$. $u_1(a)=a_1(2a_2a_3-1)+(1+\epsilon) a_5, u_2(a)=a_2(2a_1-1), u_3(a)=a_3(2a_1-1)$, $u_4(a)=a_4(2a_1a_5a_6a_7-1), u_5(a)=a_5(2\max\{a_4,a_6a_7 \}-1), u_6(a)=a_6(2\max\{a_4,a_5a_7 \}-1), u_7(a)=a_7(2\max\{a_4,a_5a_6 \}-1)$, $\epsilon \geq -1$. 

In the asynchronous game with $T=2, P_1=\{1, 4\}, P_2=\{2,3,5,6,7 \}$, if $\epsilon<0$, then the unique SPNE outcome is all players take action 1. If $\epsilon>0$, then both all players take action 1, and players 5,6,7 take action 0 and all other players take action 0 is SPNE outcome.
\end{example}

Another alternative solution concept is Markov Perfect Equilibrium(MPE).  However, the well-developed theory of MPE requires more restrictive non-degenerate conditions to avoid personalized history partition and define state variables \citep{maskin2001markov}, which cannot be satisfied as the synchronous game is in general extremely degenerated. 

\paragraph{When we have SPNE coincides with MSPNE:}

We have some positive results, that is, under additional assumptions, MSPNE coincides with SPNE for synchronous and asynchronous games. Unfortunately, we cannot find uniform conditions to guarantee it.

\begin{proposition}
\begin{enumerate}
    \item In the synchronous game, if the stage game is a weakest link game on an undirected graph such that $\forall i\in N, \forall a_{-i},a_{-i}'\in A_{-i}$, $ u_i^*(a_{-i}):=\max_{a_i\in A_i} u_i(a_i, a_{-i})$, $u_i^*(a_{-i}')=u_i(a_i=1, a_{-i}'), u_i^*(a_{-i})=u_i(a_i=0, a_{-i}) \Rightarrow f(a_{-i}')>f(a_{-i})$, then the SPNE outcome coincides with the MSPNE outcome.
    \item  In the asynchronous game,  if $\forall i\in N, \forall a_{-i},a_{-i}'\in A_{-i}$, $u_i^*(a_{-i}):=\max_{a_i\in A_i} u_i(a_i, a_{-i})$, $u_i^*(a_{-i}')=u_i(a_i=1, a_{-i}'), u_i^*(a_{-i})=u_i(a_i=0, a_{-i}) \Rightarrow f(a_{-i}')>f(a_{-i})$, then the least SPNE outcome coincides with the the least MSPNE outcome.
    \item If the stage game is a weakest link game on a directed graph, then $\forall T\in \mathbb{N^+}$, there exists a partition such that in the corresponding asynchronous game, the least SPNE outcome coincides with the least MSPNE outcome, which is further the optimal asynchronous design outcome.
    \item If the stage game is cost-ordered and contribution-ordered, then in the synchronous game the SPNE outcome coincides with the MSPNE outcome, and there exists a partition such that in the corresponding asynchronous game, the least SPNE outcome coincides with the least MSPNE outcome, which is further the optimal asynchronous design outcome.
\end{enumerate}
\end{proposition}

\paragraph{The equivalence of MSPNE between the synchronous and asynchronous games fail when the action is non-binary:}
To conclude this section, we give an example which shows that the minimum number of periods $T$ such that there exists a $T-$partition in which all players take action $1$ in all MSPNE of the corresponding asynchronous game does not coincide with the minimum number of periods $T'$ such that all players take action $1$ in all MSPNE of the corresponding monotone game in general, when the feasible action is not binary. 

\begin{example}
Consider the example with 6 players. $N=\{1,2,3,4,5,6\}$, $A_1=A_2=\{0,1,2\}, A_3=A_4=A_5=A_6=\{0,1 \}$. $u_1(a)=\mathbb{I}(a_1\geq 1)(2\mathbb{I}(a_2\geq 1)-1)+\mathbb{I}(a_1\geq 2)(2a_3a_4\mathbb{I}(a_2\geq 1)-1)$, $u_2(a)=\mathbb{I}(a_2\geq 1)(2\mathbb{I}(a_1\geq 1)-1)+\mathbb{I}(a_2\geq 2)(2a_5a_6\mathbb{I}(a_1\geq 1)-1)$, $u_3(a)=a_3(2\mathbb{I}(a_1\geq 2)-1), u_4(a)=a_4(2\mathbb{I}(a_1\geq 2)-1)$, $u_5(a)=a_5(2\mathbb{I}(a_2\geq 2)-1), u_6(a)=a_6(2\mathbb{I}(a_2\geq 2)-1)$. One can see that the unique MSPNE outcome of the monotone game is that players 1,2 take action 2, and all others take action 1, the highest action profile. However, there exists no partition such that in all MSPNE outcomes of the asynchronous game, players 1,2 take action 2, and all others take action 1. Since partition $P_1=\{1\}, P_2=\{2,3,4,5,6\}$ realizes the least MSPNE outcome that player 1 takes action 2, players 2,3,4 take action 0, and players 5,6 take action 0; similarly, partition $P_1=\{2\}, P_2=\{1,3,4,5,6\}$ realizes the least MSPNE outcome that player 2 takes action 2, players 1,5,6 take action 0, and players 3,4 take action 0. While there exists no partition such that the least MSPNE outcome coincides with the highest action profile.
\end{example}

\section{Omitted Proofs}

\paragraph{Proof of the \autoref{lem:0}.}

For the first case who has an iterated strictly dominated strategy 0, suppose not, then the player who has a strictly dominated strategy 0 at least has the incentive to deviate at $t=T$. And repeat the argument for iterated strict dominance case.

For the second case who has an iterated dominated strategy 1, suppose not and apply the single deviation principle, then for the player who has a strictly dominated strategy 1, there exists some strategy $s$ such that for a player $i$ who has a strictly dominant strategy $0$ in the stage game, he has no strict incentive to not commit at some point: $u_i(a_i=1,a_{-i}') \geq  \max_{a_i\in A_i} u_i(a_i,a_{-i}) \geq \max_{a_i\in A_i} u_i(a_i,a_{-i}=0^{N\backslash i})=u_i(a_i=0,a_{-i}=0^{N\backslash i})$, then deviation-proof condition implies $u_i(a_i=1,a_{-i}') \geq u_i(a_i=0,a_{-i}')$, which is a contradiction. And repeat the argument for iterated strict dominance case.

\paragraph{Proof of the \autoref{lem:1}.}
Single deviation principle works in this problem. Therefore, we first show that nobody has the incentive to deviate at stage $t=T$.  If $b_{iT-1}=1$, player $i$ cannot deviate at $t=T$ by the monotone constraint. Otherwise, $b_{iT-1}=0$. If $a_{iT-1}=0$, note that all other players that take strategy $g$ in the sub-game of $t=T$ under the history $\{ b_t\}_{t=1}^{T-1}$ act as if taking strategy $s$ in the sub-game of $t=T$ under the history $\{ a_t\vee b_t\}_{t=1}^{T-1}$, so given $s$ is a SPNE, player $i$ has no incentive to deviate. Otherwise, $a_{iT-1}=1$. $s$ is a SPNE implies that $u_i(a_i=1,a_{-i}')\geq \max_{a_i\in A_i} u_i(a_i,a_{-i})$ for some $a_{-i}'$ and $a_{-i}$, which is determined by strategy profile $s$. Common interests implies $\max_{a_i\in A_i} u_i(a_i,a_{-i}) \geq \max_{a_i\in A_i} u_i(a_i,0^{N\backslash i})= u_i(a_i=0,0^{N\backslash i})$. The last equation holds since action $1$ is not a strictly dominant strategy. Deviation-proof condition implies $u_i(a_i=1,a_{-i}')\geq u_i(a_i=0,a_{-i}')$. So nobody has the incentive to deviate at stage $t=T$.

Now we show that  nobody has the incentive to deviate at stage $1\leq t\leq T-1$, which is relatively easy. If $b_{it-1}=1$, player $i$ cannot deviate at $t=T$ by the monotone constraint.  Otherwise, if $a_{it}=0$, given $s$ is a SPNE, player $i$ has no incentive to deviate. Otherwise if $a_{it}=1$, $i$'s payoff is indifferent regardless of deviating or not, given all other players adopt strategy $g$. 

If $s$ is a monotonic strategy, then $g$ is a monotonic strategy, which is easy to verify by definition.

\paragraph{Proof of the \autoref{prop:1}.}

Suppose $s$ is a MSPNE of $\Gamma(T+1)$ in the form of \autoref{lem:1},  define $c_1:=0^{N},\forall 2\leq k\leq T,\forall b_k\in A_k, c_{k}:=b_{k-1}$, then the following $g$ is a MSPNE of $\Gamma(T)$: $\forall i\in N, g_{i1}=0, \forall 2\leq t\leq T-1, g_{it}(\prod_{k=1}^{t-1}b_k)=1 \text{ iff } b_{it-1}=1,  g_T(\prod_{k=1}^{T-1}b_k)=s_{T+1}(\prod_{k=1}^{T} c_k)$, since $g$ is a monotonic strategy profile, and it is simply the sub-game of $\Gamma(T+1)$ under the history $c_1=0^{N}$. 

\paragraph{Proof of the \autoref{lem:2}.}

We first show the case $T=2$. given $a_1 \in A_1$, define $P(a_1)=\{i\in N: a_{i1}=1 \vee i\in \text{ inf }O(N\backslash a_1,\overline{U},1) \}$. $P(a_1)$ is well defined since the NE is a non-empty complete lattice. We construct the following strategy $s$: $\forall i\in N, s_{i1}=0, s_{i2}(a_1)=1 \text{ iff } i\in P(a_1) $. $s$ is a monotonic strategy profile, to see it is a MSPNE, first, nobody has the incentive to deviate at $t=2$, and once $i$ deviates at $t=1$, her payoff changes from $u_i(\emptyset)$ to $u_i(P(a_{i1}=1,a_{-i1}=0))$, while $N_i$ is a MSPNE outcome implies $u_i(N_i)\geq u_i(P(a_{i1}=1,a_{-i1}=0)\cup N_i)$. Suppose $\emptyset$ is not a MSPNE outcome, then $\exists i\in N, u_i(P(a_{i1}=1,a_{-i1}=0)) > u_i(\emptyset)$, therefore, $u_i(P(a_{i1}=1,a_{-i1}=0))> u_i(P(a_{i1}=1,a_{-i1}=0)\backslash i)$ by deviation-proof condition, and $u_i(P(a_{i1}=1,a_{-i1}=0)\cup N_i)>u_i([P(a_{i1}=1,a_{-i1}=0)\cup N_i]\backslash i)$ by single-crossing condition. Suppose $u_i(N_i)\geq u_i(P(a_{i1}=1,a_{-i1}=0)\cup N_i)$, then $u_i(P(a_{i1}=1,a_{-i1}=0)\cup N_i) \geq \max\{ u_i(N_i),u_i(N_i\cup i)\}$ by common interest. Regardless of $u_i(N_i\cup i)> u_i(N_i)$ or $u_i(N_i)\geq u_i(N_i\cup i)$, common interests implies $N_i$ cannot be a MSPNE outcome, which is a contradiction.

Now we can simply extend the result to general $T$ by mathematical induction. define $P(a_1)=\{i\in N: a_{i1}=1 \vee i\in \text{ inf }O(N\backslash a_1,\overline{U},T-1) \}, \forall 2\leq k\leq T-1, P(\prod_{t=1}^{k} a_{t})=P(\prod_{t=1}^{k-1 }a_{t}) \cup \{i\in N: a_{ik}=1 \vee i\in \text{ inf }O(N\backslash (P(\prod_{t=1}^{k-1 }a_{t}) \cup a_k),\overline{U},T-k) \}$. Once we have shown the validness of the result for $T=p$, then $P(\prod_{t=1}^{p} a_{t})$ is well-defined. We construct the following strategy $s$: $\forall i\in N, s_{i1}=0, \forall 2\leq t\leq T-1, s_{it}(\prod_{k=1}^{t-1}a_{k})=1 \text{ iff } a_{it-1}=1,  s_{iT}(a_1)=1 \text{ iff } i\in P(\prod_{t=1}^{T-1} a_{t}) $. Also let me introduce the notation $P_{ik}$ to represent $P(\prod_{t=1}^{T-1} a_{t})$, in which $\forall j\in N \forall 1\leq t\leq T-1, a_{jt}=\mathbb{I}(j= i, t\geq k)$. Once we have shown the validness of the result for $T=p$, then nobody has the incentive to deviate at $t=p+1$ for the game $\Gamma(T=p+1)$. Once $i$ deviates at $t=k$, her payoff changes from $u_i(\emptyset)$ to $u_i(P_{ik})$, while $N_i$ is a MSPNE outcome implies $u_i(N_i)\geq u_i(P_{ik}\cup N_i)$, and repeat the argument above, we have $u_i(\emptyset) \geq u_i(P_{ik})$. 

\paragraph{Proof of the \autoref{prop:2}.}

$\Rightarrow: $ Suppose $\phi_{\underline{\Gamma}(N\backslash X)}(T)=Y\neq \emptyset$, then given the strategy in the form of \autoref{lem:1}, under the lowest history, those players $X$ take action $1$, then the monotonic strategy implies those players $X$ take action $1$ in any of the history, then those players $Y$ have the incentive to deviate. 

$\Leftarrow: $ Define $P(a_0)=X, \forall 1\leq k\leq T-1, P(\prod_{t=1}^{k} a_{t})=P(\prod_{t=1}^{k-1 }a_{t}) \cup \{\{i\in N: a_{ik}=1\} \cup \text{ inf }O(N\backslash (P(\prod_{t=1}^{k-1 }a_{t}) \cup a_k),\overline{U},T-k) \}$. We construct a strategy $s$ that supports it: $\forall i\in N, s_{i1}=0, \forall 2\leq t\leq T-1, s_{it}(\prod_{k=1}^{t-1}a_{k})=1 \text{ iff } a_{it-1}=1. \forall i\in X, s_{iT}=1; \forall i\in N\backslash X,  s_{iT}(a_1)=1 \text{ iff } i\in P(\prod_{t=1}^{T-1} a_{t}) $, this is simply a MSPNE by \autoref{lem:2}.

\paragraph{Proof of the \autoref{thm:1}.}

\begin{enumerate}

\item Define $t_1:=\text{ inf }\{t: \phi_{\Gamma}(t) \supset X\}$, $Y:=\phi(t_1)$ is a strictly sufficient set/equilibrium, then $\tau_{\Gamma}(X)=t_1=\tau(\underline{\Gamma}(Y))$.
 
\item $\forall Y\in SSS(\Gamma): Y\supset X$, $\tau_\Gamma(X)\leq \tau(\underline{\Gamma}(Y))$. A more formal proof: move down the player $i\in N\backslash Y$'s strategy to $s_{ik}\equiv 0, 1\leq k\leq T$, then given $T\geq \tau(\underline{\Gamma}(Y))$, those players $Y$ take action $1$ in all MSPNE, and move back those players $ N\backslash Y$' strategies, strategic complementarities and monotonic-strategy implies the equilibrium profile cannot be lower, which proves the desired result.

\item $\tau(\Gamma)=\tau(\overline{\Gamma}(N\backslash i)) \text{ if } u_i(i)>u_i(\emptyset)$, this is obvious. $\tau(\Gamma)=1 \text{ if } N=\emptyset$, this is obvious too.

\item $\forall i\in N, \tau(\Gamma) \leq 1+ \tau(\overline{\Gamma}(N\backslash i))$. To show that, suppose given $T=1+ \tau(\overline{\Gamma}(N\backslash i))$, in some MSPNE $s$, the corresponding outcome $X\neq N$, if $i\notin X$, then obviously $i$ has the incentive to deviate to $s_{i1}'=1$; otherwise, the least outcome of the game cannot be lower than $O(\overline{\Gamma}(N\backslash i),T-1)$. In both cases, $N$ is the unique feasible MSPNE outcome.

\item $\forall X\in SSS(\Gamma)\backslash \{\emptyset, N\}, \tau(\Gamma)\leq \text{max}\{ \tau(\underline{\Gamma}(X)),\tau(\overline{\Gamma}(N\backslash X)) \}$. A similar argument as above shows that $\tau_{\Gamma}(X) \leq \tau(\underline{\Gamma}(X))$, and given $T\geq \tau(\underline{\Gamma}(X)$, the previous argument we use shows that those players $N\backslash X$ can simplify the game into $\overline{\Gamma}(N\backslash X)$, so $\tau_{\Gamma}(N\backslash X)\leq \text{max}\{ \tau(\underline{\Gamma}(X)),\tau(\overline{\Gamma}(N\backslash X)) \}$, then $\tau(\Gamma)\leq \text{max}\{ \tau(\underline{\Gamma}(X)),\tau(\overline{\Gamma}(N\backslash X)) \}$.

\item $\tau_{\Gamma}(N)\leq \tau(\Gamma)$. the first two cases are obvious; a similar argument as above shows that $\forall X\in SSS(\Gamma)\backslash \{N, \emptyset\}, \tau_{\Gamma}(N) \leq max\{ \tau(\underline{\Gamma}(X)),\tau(\overline{\Gamma}(N\backslash X)) \}$. To show that $\forall i\in N, \tau_{\Gamma}(N) \leq  1+ \tau(\overline{\Gamma}(N\backslash i))$, 
    
\item Define $t_1:=\text{ inf }\{t: \phi_{\Gamma}(t) \neq \emptyset \}$, if $X_1:=\phi(t_1)\neq N$, then $\tau_(\Gamma)=\tau(\overline{\Gamma}(N\backslash X_1))> \tau(\underline{\Gamma}(X_1)) =t_1$. To see $\forall i\in N, 1+ \tau(\overline{\Gamma}(N\backslash i))\geq \tau(\overline{\Gamma}(N\backslash X_1))$, note that $1+ \tau(\overline{\Gamma}(N\backslash i))\geq \tau(\Gamma)$ in above. Therefore, $\tau(\Gamma)=\text{min} \{ 1+ \text{min}_{i\in N} \tau(\overline{\Gamma}(N\backslash i))$, \\
$min_{X\in SSS(\Gamma)\backslash \{\emptyset, N\} } max\{ \tau(\underline{\Gamma}(X)),\tau(\overline{\Gamma}(N\backslash X)) \}\}$ if $X_1:=\phi(t_1)\neq N$.

\item Otherwise, $\phi(t_1)= N$. Then $T<\tau_{\Gamma}(N) \Leftrightarrow \phi_{\Gamma}(T)=\emptyset$ by definition. We show that if \small{$T=\text{min} \{ \text{min}_{i\in N} \tau(\overline{\Gamma}(N\backslash i)), min_{X\in SSS(\Gamma)\backslash \{\emptyset, N\} } max\{ \tau(\underline{\Gamma}(X)),\tau(\overline{\Gamma}(N\backslash X)) \}-1\}$}, then $\phi_{\Gamma}(T)=\emptyset$ by constructing the most conservative strategy in \autoref{prop:2}. Note that no player has the incentive to deviate: if $i$ deviates at $t=k$, $P_{ik}\neq N$ since $T-1 < \tau(\overline{\Gamma}(N\backslash i))$. Then if $u_i(P_{ik})> u_i(\emptyset)$, $P_{ik}$ must be a strictly sufficient set/equilibrium. $T< max\{ \tau(\underline{\Gamma}(P_{ik})),\tau(\overline{\Gamma}(N\backslash P_{ik})) \}$ implies that, $\emptyset \subsetneqq P_{ik}\subset \phi_{\Gamma}(T)\subsetneqq N$, which violates the assumption that $\phi(t_1)= N$, therefore, is a contradiction. 
\end{enumerate}

\paragraph{Proof of the \autoref{lemma1}:}
\begin{proof}
For a 2-partition game, the action profiles by $N_2$ is obviously a NE of the game $\Gamma(N_2,A_{N_2},U'(a_{N_2}\cup a_{N_1}))$. Suppose some player $i\in N_1$ has strict incentive to deviate from $a_i$ to $a_{i}'$ in the stage game, i.e., $u_i(a_i,a_{N_1\backslash i},s_{N_2}(a_i,a_{N_1\backslash i}))<u_i(a_i',a_{N_1\backslash i},s_{N_2}(a_i,a_{N_1\backslash i}))$, while $a$ is a SPNE outcome implies that $u_i(a_i,a_{N_1\backslash i},s_{N_2}(a_i,a_{N_1\backslash i}))\geq u_i(a_i',a_{N_1\backslash i},s_{N_2}(a_i',a_{N_1\backslash i}))$. If $a_i>a_i'$, since $s$ is monotonic in history, $s_{N_2}(a_i,a_{N_1\backslash i}) \geq s_{N_2}(a_i',a_{N_1\backslash i})$, the deviation-proof condition implies that $u_i(a_i,a_{N_1\backslash i},s_{N_2}(a_i,a_{N_1\backslash i}))\geq u_i(a_i',a_{N_1\backslash i},s_{N_2}(a_i,a_{N_1\backslash i}))$, which is a contradiction. Else if $a_i<a_i'$, then single-crossing condition shows that \\ $u_i(a_i,a_{N_1\backslash i},s_{N_2}(a_i',a_{N_1\backslash i}))<u_i(a_i',a_{N_1\backslash i},s_{N_2}(a_i',a_{N_1\backslash i}))$. Then the common interest condition implies that $u_i(a_i',a_{N_1\backslash i},s_{N_2}(a_i',a_{N_1\backslash i}))\geq u_i(u_i(a_i',a_{N_1\backslash i},s_{N_2}(a_i,a_{N_1\backslash i})))$, which is larger than $u_i(a_i,a_{N_1\backslash i},s_{N_2}(a_i,a_{N_1\backslash i})))$, a contradiction. 

For a $k-$partition game, repeating the argument above yields the desired result.  
\end{proof}

\paragraph{Proof of the \autoref{theorem1}:}
\begin{proof}

$1\rightarrow 2\rightarrow 3\rightarrow 1$.

$1\rightarrow 2$: What we need to do is to construct a MSPNE $s$ such that $a^*$ is the corresponding MSPNE outcome. The proof is almost self-contained in the definition: $\forall i\in N_T, \forall h_{T-1}\in H_{T-1}, s_i(h_{T-1})=\hat{a_{T}}(h_{T-1})$, the least NE of the auxiliary game $\Gamma(N_T,A_T,U'(h_{T-1}))$, which is a part of MSPNE. $\forall i\in N_{T-1}, \forall h_{T-2}\in H_{T-2}, s_i(h_{T-2})=\hat{a_{T}}(h_{T-2})$, the least NE of the auxiliary game $\Gamma(N_{T-1},A_{T-1},U'(h_{T-2}))$, which is a part of MSPNE given those players in $N_T$'s strategies fixed. Repeating the argument yields the desired result.

$2\rightarrow 3$: Without loss of generality, we may assume that in the least action profile survives IESEDS, the set of players who take action 1 is $M$. Without loss of generality, we can further assume $M=N$, since we can otherwise consider the new stage game $\Gamma(M,A_M,U')$, in which $U': A_M\rightarrow \mathbb{R}^{|M|}$. $\forall i\in M, a_M\in A_M, U_i'(a_M)=u_i(a_M,a_{N\backslash M}=0^{N\backslash M})$. Given the construction of $a^*$, $\exists i\in N_1$ has a strictly dominated strategy $0$ in the corresponding game $\Gamma(N_1,A_1,U')$, in which $\forall i\in N_1, a_1\in A_1, u_i'(a_1)=u_i(a_1,a_2^*(a_1))$, i.e., $u_i(a_i=1,a_{N_1\backslash i}=0^{N_1\backslash i},a_2^*(a_i=1,a_{N_1\backslash i}=0^{N_1\backslash i}))>u_i(a_1=0^{N_1},a_2^*(a_{N_1}=0^{N_1}))$. Then Let $B_i:=\{j\in N\backslash N_1: a_2^*(a_i=1,a_{N_1\backslash i}=0^{N_1\backslash i})|_{j}=1 \}$. $V(i):=\cup_{j\in B_i\cup i}\cup_{k\in N\backslash (B_i\cup i)}\{(j,k) \} \bigcup \cup_{j\in B_i} \{(j,i) \}$. Now consider the auxiliary game $\Gamma(B_i,A_{B_i},U')$, in which $\forall j\in B_i, \forall a_{B_i}\in A_{B_i}$, $u_j'(a_{B_i}):=u_j(a_{B_i},a_i=1,a_{N\backslash (B_i\cup i)}=0^{N\backslash (B_i\cup i)})$. According to the construction of $a^*$ for game $\Gamma(N,A,U)$, we know that $a^*=1^{B_i}$ for game $\Gamma(B_i,A_{B_i},U')$. So we can again pick some $j\in B_i\cap N_2$, and the corresponding $B_{j}$ for the auxiliary game $\Gamma(B_i,A_{B_i},U')$, and let $V(i,j):=V_i\bigcup \cup_{l\in B_j\cup j}\cup_{k\in B_i\backslash (B_j\cup j)}\{(l,k) \} \bigcup \cup_{l\in B_j} \{(l,j) \}$. Repeating the argument derives the desired result. Not hard to see that the graph we construct is $M-$feasible.

$3\rightarrow 1$: Without loss of generality, we only consider the greatest feasible graph, which means $\nexists P$-feasible graph such that $P\backslash M\neq \emptyset$. Otherwise, we may maintain the graph structure of $M$ and $P\backslash M$ unchanged, and $V':=V\bigcup \cup_{i\in M}{j\in P\backslash M}\{(i,j) \}$, which is therefore $M\cup P$-feasible. Without loss of generality, we can further assume $M=N$, since we can otherwise consider the new stage game $\Gamma(M,A_M,U')$, in which $U': A_M\rightarrow \mathbb{R}^{|M|}$. $\forall i\in M, a_M\in A_M, U_i'(a_M)=u_i(a_M,a_{N\backslash M}=0^{N\backslash M})$. Now we prove it with mathematical induction. The result is true for $T=1$ since the existence of such a consistent graph implies that the unique action profile that survives IESDS is all players take action 1. Suppose the claim holds for $T=k$. Now for $T=k+1$, given the graph $G$ is consistent, either $G$ is strongly connected and $|N_1|=1$, or there are several strongly connected components. So we can always pick some $i\in N_1$ that is unreachable from other $j\in N_1$. We claim that player $i$ has a strict incentive to take action $1$ in all MSPNE: according to the induction assumption, all those players that are in the same connected component of player $i$, represented by $P_i$, will take action 1 in all MSPNE. Given all players' strategies are monotonic in history, player i's payoff from taking $1$ is $u_i(a_i=1, a_{P_i}=1^{P_i}, a'_{-P_i\cup i})$, which is larger than $u_i(a_i=0, a_{P_i}'=1^{P_i}, a'_{-P_i\cup i})$ since $G$ is consistent. player i's payoff from taking $0$ is $u_i(a_i=0, a_{P_i}, a_{-P_i\cup i})$, which is less than $u_i(a_i=1, a_{P_i}=1^{P_i}, a'_{-P_i\cup i})$ due to the common interest assumption. So all those players in the same strongly connected component of player $i$ will take action 1 for sure. Repeating the argument above proves the desired result. Mathematical induction then implies the validity of the claim.

\end{proof}

\paragraph{Proof of the \autoref{proposition1}:}
\begin{proof}
The proof relies on the equivalent graph representation. Suppose there exists a $M_1$-feasible graph $G_1=(M_1,V_1)$ with respect to partition $P_1$, and there exists a $M_2$-feasible graph $G_2=(M_2,V_2)$ with respect to partition $P_2$. We claim that there exists a $M_1\cup M_2$-feasible graph with respect to a specific partition. We construct the partition in the following way: fix all those players in $M_1$'s position according to partition $P_1$, and fix all those players in $M_2 \backslash M_1$'s position according to partition $P_2$. We construct the $M_1\cup M_2$-feasible graph with respect to the previous partition: $V=V_1\bigcup V_2|_{M_2\backslash M_1} \bigcup \cup_{i\in M_1}\cup_{j\in M_2\backslash M_1}\{i,j\}$. The constructed graph $V$ is $M_1\cup M_2$ feasible, which proves the first part of the proposition. For the second part, find the smallest $1\leq t\leq T$ such that $|N_t|\geq 2$, arbitrarily choose some $i\in N_t$, and construct $N_l'=N_l$ for $l\leq t-1$, $N_t'=\{i\}$, $N_{t+1}'=N_t\backslash i$, $N_{l+1}'=N_l$ for $t+1\leq l \leq T$. The constructed graph is $M$-feasible with respect to the partition $\{ N_t'\}_{t'=1}^{T+1}$ then.
\end{proof}

\paragraph{Proof of the \autoref{lemma2}:}
\begin{proof}
The proof is basically self-contained. For the first part, all nodes in $N_k$ belong to different strongly connected components in $G|_{\cup_{t\geq k}N_t}$. For the second part, see tree-depth as a value function and consider the corresponding policy function. If the policy function is to divide the graph into several strongly connected components, then consider the position of nodes in each of the components independently. If the policy function is to remove node $i$, then $N_1=\{i\}$ and consider $N_2$ next.
\end{proof}

\paragraph{Proof of the \autoref{prop:4}.}
\begin{enumerate}
    \item The first case is trivial.
    \item Otherwise if $G$ is strongly connected, then $SSS(\Lambda)=\{N\}$, and apply \autoref{thm:1}.
    \item Otherwise, there exists at least one strongly connected component $G_i=(N_i,E_i)$, $G_i:=G|_{N_i}$ such that $W(N_i)=N_i$, in which the game is isolated, so $\tau_{\Lambda}(N_i)=\tau(\Lambda(N_i;G_i))$.
    \item For the remaining graph $G|_{N\backslash N_i}$, there exists at least one strongly connected component $G_j=(N_j,E_j)$ such that $W(N_i)=N_i$. Then we must have $\tau_{\Lambda}(N_i\cup N_j)= \text{max}\{\tau(\Lambda(N_i;G_i)), \tau(\Lambda(N_j;G_j))\}$.
    \item Repeat the process, we have $\tau_{\Lambda}(N)=\text{max}_{i}\tau(\Lambda(N_i;G_i))$, where $G_i$ is the strongly connected component.
    \item We rename $\tau(\Lambda(N_i;G_i)):=td(G_i)$, which is the desired result. 
    \item Note that the tree-depth(td) is weakly increasing in the graph $G$, so apply \autoref{thm:1} and we get $\tau_{\Lambda}(X)=td(G|_{W(X)})$.
\end{enumerate}

\paragraph{Proof of the \autoref{thm:2}:}
 \begin{proof}
 \begin{enumerate}
     \item A constructive proof. We will use the three operations defined above: dominate, divide, and delete. The priority is: divide, dominate, and delete, given the policy functions may admit several operations.
     \item Consider the smallest $t_1$ such that $\phi_{\Gamma}(t_1)\neq \emptyset$, if $\phi_{\Gamma}(t_1) \neq N$, then divide the players $N$ into $N_1:=\phi_{\Gamma}(t_1)$ and $N_2:=N\backslash \phi_{\Gamma}(t_1)$, $\forall i\in N_1,j\in N_2$, $i\rightarrow j$. Reconsider $\underline{\Gamma}(N_1)$, $\overline{\Gamma}(N_2)$, $G|_{N_1}$ and $G|_{N_2}$.
     \item Else if the operation is to divide the players into $N_1$ and $N_2$ by the policy function in \autoref{thm:1}, $\forall i\in N_1,j\in N_2$, $i\rightarrow j$. Reconsider $\underline{\Gamma}(N_1)$, $\overline{\Gamma}(N_2)$, $G|_{N_1}$ and $G|_{N_2}$.
    \item Else if dominate $i$, $\forall j\in N\backslash i, i\rightarrow j$. Reconsider $\overline{\Gamma}(N\backslash i)$ and $G|_{N \backslash i}$.
     \item Else if delete $i$, $\forall j\in N\backslash i, i\leftrightarrow j$. Reconsider $\overline{\Gamma}(N\backslash i)$ and $G|_{N \backslash i}$.
 \end{enumerate}

It is easy to check that the graph we constructed is a sufficient network of $\Gamma$, one can simply remove some branches to make it minimal. 

On the one hand, $\forall X\in 2^N$, the graph $G$ we constructed can make sure that $G|_{W(X)}$ be eliminated in at most $k=\tau_{\Gamma}(X)$ steps, (since we always take the division operation first), which is obvious from the procedure we construct it. so $\tau_{\Gamma}(X) \geq td(G|_{W(X)})$. On the other hand, $\tau_{\Gamma}(X)\leq td(G|_{W(X)})$ by \autoref{lem:3}.
 \end{proof}

\paragraph{Proof of the \autoref{prop:8}.}
\begin{proof}
The proof of the first part is easy. Given the game is cost-ordered and contribution-ordered, it implies that the former players are easier to dominate, while deleting the latter players contributes more to other players. Therefore, the optimal policy is to always delete the last players. To see why dividing the game into two sub-games does not accelerate, it is enough to note that: on the one hand, for the first sub-game $\underline{\Gamma}(X)$, compared with deleting $|X|$, deleting $|N|$ always contributes more for all players in $X$; on the other hand, given $T \geq \underline{\Gamma}(X)$, we must have that all those players in $X$ have been either dominated or deleted under the new policy, so dividing the game into two sub-games does not make additional benefit. 

The validness of the algorithm when the stage game is both strictly cost-ordered and contribution natural is guaranteed by the following fact:

$\tau_{\Gamma}(X)=\tau(\underline{\Gamma}(l(X)))$, $l(X)$ is the least element in $SSE(\Gamma)$ such that $X\subset l(X)$. 

\begin{footnotesize}
    
   \begin{equation}	  
\tau(\Gamma)=
 		\begin{cases}
 \tau(\overline{\Gamma(N\backslash 1)}) & \text{if } u_1(1)>u_1(\emptyset)\\
    1 & \text{else if } N=\emptyset\\
    1+  \tau(\overline{\Gamma(N\backslash |N|)}) & \text{else if $SSS(\Gamma)=\emptyset$}
 		\end{cases}
 	\end{equation} 
\end{footnotesize}

Since such a stage game is $\tau-$equivalent to an aggregative game with heterogeneous cost, it is enough to focus on the number of players who take action $1$, which accelerates the computation. Instead of computing the sub-game $\underline{\Gamma}(l(X))$, again, we may solve the problem equivalently during computing $\tau(\Gamma)$ since whether a player has the incentive to take action $1$ strictly depends on the number of players who take action $1$ only.
\end{proof}

\end{document}